\documentclass[twoside,leqno,twocolumn]{article}

\usepackage[letterpaper]{geometry}
\usepackage{xspace}
\usepackage{ltexpprt}
\usepackage{graphicx}
\usepackage{color}
\usepackage{amsmath}
\usepackage{amssymb}
\usepackage{enumitem}
\usepackage{cases}
\usepackage{amstext}
\usepackage{float}
\usepackage{thm-restate}
\usepackage{tikz}
\usetikzlibrary{automata,positioning}
\usetikzlibrary{calc, fit, shapes,intersections,arrows,matrix, topaths}
\usetikzlibrary{decorations,decorations.pathreplacing,patterns, shadows}
\usepackage{mathtools}
\usepackage{microtype}
\usepackage{thmtools}
\newcommand{\cont}[2]{$\left\{ \begin{minipage}{.9cm} .. #1 \\ .. #2
\end{minipage}\right.$}
\newcommand{\contt}[2]{$\left\{ \begin{minipage}{1.2cm} .. #1 \\ .. #2
\end{minipage}\right.$}

\newcommand{\defn}[1]       {{\textit{\textbf{#1}}}}
\newfont{\subheadfnt}{ptmbi8t at 10pt}

\renewcommand{\paragraph}[1]{%
{\normalfont\normalsize\bfseries #1}%
}%

\newcommand{\thmref}[1]     {Theorem~\ref{thm:#1}}
\newcommand{\secref}[1]     {Section~\ref{sec:#1}}
\newcommand{\appref}[1]     {Appendix~\ref{app:#1}}

\newcommand{\figref}[1]     {Figure~\ref{fig:#1}}
\newcommand{\defref}[1]     {Definition~\ref{def:#1}}
\newcommand{\lemref}[1]     {Lemma~\ref{lem:#1}}

\newcommand{\floor}[1]      {\left\lfloor #1 \right\rfloor}

\newtheorem{thm}{Theorem}
\newtheorem{prop}[thm]{Proposition}
\newtheorem{definition}[thm]{Definition}
\newcommand{\proporef}[1]         {Proposition~\ref{propo:#1}}
\newcommand{\proporeftwo}[2]         {Propositions~\ref{propo:#1} and ~\ref{propo:#2}}

\numberwithin{lemma}{section}
\numberwithin{definition}{section}
\numberwithin{prop}{section}
\newcommand{\pagereq}[1]{\ensuremath{\sigma_{i, #1}}\xspace}
\newcommand{\pageij}{\pagereq{j}}

\newcommand{\costprop}{cost property\xspace}
\newcommand{\ccp}{bounded-shared-\costprop}
\newcommand{\Ccp}{Bounded-shared-\costprop}
\newcommand{\cgraph}{input-cost graph\xspace}
\newcommand{\mc}{multicore\xspace}
\newcommand{\MC}{Multicore\xspace}
\newcommand{\paging}{caching\xspace}
\newcommand{\mcp}{\mc \paging}
\newcommand{\MCp}{\MC \paging}
\newcommand{\ts}{timestep\xspace}
\newcommand{\tss}{timesteps\xspace}
\newcommand{\alg}{\ensuremath{\mathcal{A}}\xspace}

\newcommand{\opt}{\textsc{OPT}\xspace}
\newcommand{\FWF}{\textsc{FWF}\xspace}

\newcommand{\req}{\ensuremath{\mathcal{R}}\xspace}

\newcommand{\ma}{\ensuremath{\mathcal{A}}\xspace}
\newcommand{\mb}{\ensuremath{\mathcal{B}}\xspace}
\newcommand{\sked}{\ensuremath{\mathcal{S}}\xspace}
\newcommand{\skedar}{\ensuremath{\sked_{\req, \ma}}\xspace}
\newcommand{\skedbr}{\ensuremath{\sked_{\req, \mb}}\xspace}

\newcommand{\fif}{\textsc{FIF}\xspace}

\newcommand{\lru}{\textsc{LRU}\xspace}
\newcommand{\LRU}{\textsc{LRU}\xspace}
\newcommand{\FIFO}{\textsc{FIFO}\xspace}
\newcommand{\ith}{\ensuremath{i}th\xspace}
\newcommand{\core}{\ensuremath{P}\xspace}

\newcommand{\seqset}{\ensuremath{\mathcal{I}}\xspace}
\newcommand{\localset}{\ensuremath{\mathcal{I}^f}\xspace}

\newcommand{\seqsetnt}{\ensuremath{\mathcal{I}(n^t_1, \ldots, n^t_p)}\xspace}

\newcommand{\sleq}{\ensuremath{\preceq_s}\xspace}

\newcommand{\bequiv}{\ensuremath{\equiv_c}\xspace}
\newcommand{\bleq}{\ensuremath{\preceq_c}\xspace}
\newcommand{\bless}{\ensuremath{\prec_c}\xspace}
\newcommand{\bnleq}{\ensuremath{\npreceq_c}\xspace}
\newcommand{\surj}{natural\xspace}
\newcommand{\surjcap}{Natural\xspace}
\newcommand{\skedpre}{\ensuremath{\skedar^{\text{pre}}}\xspace}
\newcommand{\skedsuf}{\ensuremath{\skedar^{\text{suf}}}\xspace}
\newcommand{\inp}{input\xspace}
\newcommand{\inps}{inputs\xspace}
\newcommand{\skedprej}{\ensuremath{\sked_{j, \req, \ma}^{\text{pre}}}\xspace}
\newcommand{\skedsufj}{\ensuremath{\sked_{j, \req, \ma}^{\text{suf}}}\xspace}
\newcommand{\reqpre}{\ensuremath{\req^{\text{pre}}}\xspace}
\newcommand{\reqsuf}{\ensuremath{\req^{\text{suf}}}\xspace}
\newcommand{\pnlru}{\ensuremath{\sigma_{\text{NLRU}}}\xspace}
\newcommand{\plru}{\ensuremath{\sigma_{\text{LRU}}}\xspace}
\newcommand{\compsuf}{\ensuremath{\overline{\reqsuf}}\xspace}
\newcommand{\curr}{\ensuremath{\reqpre r_{j+1}}\xspace}
\newcommand{\skedloc}{\ensuremath{\sked_{\localset,\ma}}\xspace}

\newcommand{\nar}{\ensuremath{n_{\req, \ma}}\xspace}
\newcommand{\fullcompsuf}{\ensuremath{\reqpre r_{j+1} \compsuf}\xspace}

\newcommand{\skedmapari}{\ensuremath{\sked_{\pi(\req_i), \ma}}\xspace}
\newcommand{\skedbri}{\ensuremath{\sked_{\req_i, \mb}}\xspace}
\newcommand{\inv}{\ensuremath{\mathcal{V}}\xspace}
\newcommand{\invpage}{\ensuremath{\inv_{\sigma, \req, \ma, \mb}}\xspace}
\newcommand{\invpageset}{\ensuremath{\inv_{\sigma, \localset, \ma, \mb}}\xspace}
\newcommand{\invnlru}{\ensuremath{\inv_{\pnlru, \req, \ma, \mb}}\xspace}
\newcommand{\invpagesetnlru}{\ensuremath{\inv_{\pnlru, \localset, \ma,
      \mb}}\xspace}
\newcommand{\univ}{\ensuremath{\mathcal{I}}\xspace}

\newcommand{\hass}{Schedule-Explicit\xspace}

\newcommand{\newbacap}{Cyclic\xspace}
\newcommand{\newba}{cyclic\xspace}
\newcommand{\ba}{bijective analysis\xspace}
\newcommand{\bawithshort}{analysis\xspace}
\newcommand{\fullnewba}{\newba \bawithshort}
\newcommand{\fullnewbacap}{\newbacap \bawithshort}
\newcommand{\shortnewba}{cyclic analysis\xspace}

\newcommand{\eg}{e.g.,\xspace}
\newcommand{\lazy}{lazy\xspace}

\newcommand{\lm}{lazy\xspace}
 
\providecommand{\qedsymbol}{$\square$}
\DeclareRobustCommand{\qed}{%
  \ifmmode 
  \else \leavevmode\unskip\penalty9999 \hbox{}\nobreak\hfill
  \fi
  \quad\hbox{\qedsymbol}}
\newenvironment{proofsketch}	{\noindent\textsc{Proof Sketch.}\hspace{0.5em}}{\qed}

\begin{document}

\title{Beyond Worst-case Analysis of Multicore Caching Strategies}

\author{Shahin Kamali\thanks{Department of Computer Science,
    University of Manitoba, Winnipeg, Mb, Canada.
    Email:~\texttt{shahin.kamali@umanitoba.ca.}}
  \and
  Helen Xu\thanks{Computer Science and Artificial Intelligence Laboratory,
    Massachussetts Institute of Technology, Cambridge MA 02139
    USA. Email:~\texttt{hjxu@mit.edu}.  }
  }
\date{}

\maketitle







  Every processor with multiple cores sharing a cache needs to implement a
  cache-replacement algorithm. Despite the widespread use of multicore systems, our
  theoretical understanding of caching algorithms in the multicore setting lags far
  behind our understanding of the single-core setting. 
  In this paper, we consider the free-interleaving model for multicorep, where a set
  of cores issues online requests to a shared cache. If a requested page by a
  core $p$ is in the cache, $p$ can issue another request in the next
  timestep. Otherwise, it takes $\tau>1$ timesteps to fetch the requested page
  to the cache before $p$ can make its next request. Meanwhile, other cores can
  continue issuing requests while $p$ is fetching. Hence, a page-replacement
  policy implicitly defines a schedule that defines how requests are served at
  any time-step.

  Previous work demonstrated that the competitive ratio of a large class of
  online algorithms, including Least-Recently-Used (LRU), grows with the length
  of the input. Furthermore, even offline algorithms like Furthest-In-Future,
  the optimal algorithm in single-core caching, cannot compete in the multicore setting.
  These negative results 
  motivate a more in-depth comparison of multicore caching algorithms via alternative
  analysis measures. Specifically, the power of the adversary to adapt to online
  algorithms suggests the need for a direct comparison of online algorithms to
  each other.

  In this paper, we introduce cyclic analysis, a generalization of bijective analysis introduced
  by Angelopoulos \textit{et al.}  [SODA'07]. Cyclic analysis captures the
  advantages of bijective analysis while offering flexibility that makes it more useful for comparing algorithms for a variety online problems. In particular, we take the first steps beyond worst-case analysis for analysis of multicorep algorithms.  
We use cyclic analysis to establish
  relationships between multicorep algorithms, including the advantage of LRU over all
  other multicore caching algorithms in the presence of locality of reference.

\newcommand{\ssk}[1]{{\footnotesize \sf {\color{magenta}{#1}}}}

\section{Introduction}\label{sec:intro}
Despite the widespread use of multiple cores in a single machine, the
theoretical performance of even the most common cache eviction algorithms is not
yet fully understood when multiple cores simultaneously share a cache. Caching
algorithms for \mc architectures have been well-studied in practice, including
dynamic cache-partitioning heuristics~\cite{QurePa06, StoneTuWo92, SuhRuDe04}
and operating system cache
management~\cite{QureJaPa07,XieLo09,FedoSeSm06}. There are very few theoretical
guarantees, however, for performance of these algorithms. Furthermore, most
existing guarantees on online \mcp algorithms are
negative~\cite{LopezSa12a,KamaliXu20}, but resource augmentation may be helpful
in some cases~\cite{AgraBeDa20, AgrawalBDKPS21}.

In this paper, we explore the \defn{\mcp\footnote{This problem is also called
    ``paging'' in the literature~\cite{LopezSa12a}. We use ``\mcp'' because
    it more accurately reflects the problem studied in this paper.} problem} in which
multiple cores share a cache and request pages in an online manner. Upon serving
a request, the requested page should become available in the shared cache. If
the page is already in the cache, a hit takes place; otherwise, when the page is
not in the cache, the core that issues the request incurs a miss. In case of a
miss, the requested page should be fetched to the cache from a slow
memory. Fetching a page causes a \defn{fetch delay} in serving the subsequent
requests made by the core that incurs the miss. Such delay is captured by the
free-interleaving model of \mcp~\cite{LopezSa12a,
  KattiRa12}. 
Under this model, when a core incurs a miss, it spends multiple cycles fetching
the page from the slow memory while other cores may continue serving their
requests in the meantime. Therefore, an algorithm's eviction strategy not only
defines the state of the cache and the number of misses, but also the order in
which requests are served. That is,
a \paging algorithm implicitly defines a
``schedule'' of requests served at each timestep through its previous
eviction decisions.
%

\textbf{Divergence between \mc and single-core \paging.} Previous
work~\cite{LopezSa12a, KamaliXu20} leveraged the scheduling aspect of \mcp to
demonstrate that guarantees on competitive ratio\footnote{For a
  cost-minimization problem, an online algorithm has a competitive ratio of $c$
  if its cost on any input never exceeds $c$ times the cost of an optimal
  offline algorithm for the same input (up to an additive constant).} of algorithms
in the single-core setting do not extend to \mcp. In particular, L\'{o}pez-Ortiz
and Salinger~\cite{LopezSa12a} focused on two classical single-core \paging
algorithms, \textsc{Least-Recently-Used} (\lru)~\cite{SleaTa85} and
\textsc{Furthest-In-Future} (\fif)~\cite{Belady66}, and showed these algorithms
are unboundedly worse than the optimal algorithm \opt in the free-interleaving
model\footnote{\lru is an online \paging algorithm that evicts the
  least-recently-requested page. \fif is an offline \paging algorithm that
  evicts the page that will be requested furthest in the future.  Both
  algorithms evict pages only when the cache is full and there is a request to a
  page not in the cache. In the \mc setting, ties can happen; both \lru and \fif
  break ties arbitrarily.}. In the free-interleaving model, \fif evicts the page
furthest in the future in terms of the number of requests. In the single-core
setting, \lru is $k$-competitive (where $k$ is the size of the
cache)~\cite{SleaTa85}, and \fif is the optimal
algorithm~\cite{Belady66}. Kamali and Xu~\cite{KamaliXu20} further confirmed the
intuition that \mc \paging is much harder than single-core \paging and showed
that all \lm algorithms are equivalently non-competitive against \opt. An online
\paging algorithm is \defn{\lm}~\cite{ManaMcSl90} if it 1) evicts a page only if
there is a miss 2) evicts no more pages than the misses at each timestep, 3) in
any given timestep, does not evict a page that incurred a hit in that timestep,
and 4) evicts a page only if there is no space left in the cache\footnote{Lazy
  algorithms are often called ``demand paging'' in the systems
  literature~\cite{Salt74}. Algorithms with properties 1-3 (but not necessarily
  4) are called ``honest'' algorithms~\cite{LopezSa12a}.}.  Lazy algorithms
capture natural and practical properties of online algorithms. Common \paging
strategies such as \lru and First-In-First-Out (\FIFO) are clearly
lazy. Unfortunately, the competitive ratio of this huge class of algorithms is
bounded and grows with the length of the input.

The existing negative results for competitive analysis consider a cost model in
which the goal is to minimize the number of misses. Nevertheless, they extend to
the case when the objective is the total number of \tss to answer all
requests~\cite{KamaliXu20}. We focus on the latter measure in this paper because
it is more practical, as we will explain in detail in~\secref{prelim}.

At a high level, the divergence between performance of algorithms for \mc and
single-core \paging stems from the power of the adversary to adapt to online
algorithms and to generate inputs that are particularly tailored to harm the
schedule of online algorithms. For these adversarial inputs, the implicit
scheduling of lazy algorithms causes periods of ``high demand'' in which the
cache of the algorithm is congested (cores request many different pages).
Meanwhile, an optimal offline algorithm avoids these high-demand periods by
delaying cores in an ``artificial way''. These adversarial inputs highlight the
inherent pessimistic nature of competitive analysis. \vspace*{1mm}

\textbf{Beyond worst-case analysis.}  The highly-structured nature of the
worst-case inputs suggests that competitive analysis might not be suitable for
studying \mc \paging algorithms and motivates the study of alternatives to
competitive ratio. There are two main reasons to go beyond competitive analysis
for analysis of \mc \paging algorithms. First, competitive analysis is overly
pessimistic and measures performance on worst-case sequences that are unlikely
to happen in practice. In contrast, measures of typical performance are more
holistic than worst-case analysis, which dismisses all other sequences. Second,
competitive analysis does not help to separate online algorithms for \mcp because
no practical algorithm can compete with an optimal offline
algorithm~\cite{KamaliXu20}. Therefore, other measures are required to establish
the advantage of one online algorithm over others.  Many alternative measures
have been proposed for single-core \paging~\cite{Young98, Young94,
  KoutsoupiasPa00,KarlinPhRa92,BoyarLaNi01,BoyarFaLe05,Young02,Young00,BenBo94}. For
a survey of measures of online algorithms, we refer the reader
to~\cite{DorrigivLo05, Komm16, BoyarFaLa18}.
%
In particular, bijective analysis~\cite{AngeDoLo08,AngeSc09,AngeSc13} is a natural
measure that directly compares online algorithms and has been used to capture
the advantage of \lru over other online single-core \paging algorithms on inputs
with ``locality of reference''~\cite{AngeDoLo08,AngeSc09}. Despite these
results, as we will show, \ba has restrictions when it comes to \mcp.

\subsection{Contributions} We take the first steps beyond competitive analysis
for \mc \paging by extending \ba to a stronger measure named \fullnewba and
demonstrating how to apply \fullnewba to analyze \mcp algorithms. The
pessimistic nature of competitive analysis demonstrates the need for alternative
measures of online algorithms.

\vspace*{1mm} \textbf{\fullnewbacap.} We introduce \shortnewba, a measure that
captures the benefits of bijective analysis and offers additional flexibility
which we will demonstrate in our analysis of \mc \paging
algorithms. \fullnewbacap generalizes bijective analysis by directly comparing
two online algorithms over \emph{all inputs}.  Traditional bijective analysis
compares algorithms by partitioning the universe of \inps based on \inp length
and drawing bijections between \inps in the same
partition~\cite{AngeDoLo07,AngeSc09, AngeDoLo08,Dorr10}.  \fullnewbacap relaxes
this requirement by allowing bijections between \inps of different lengths. This
flexibility allows for alternative proof methods for showing relationships
between algorithms.

We show that all \lm~\cite{LopezSa12a, ManaMcSl90} algorithms are equivalent
under cyclic analysis. More interestingly, we show the strict advantage of any
lazy algorithm over Flush-When-Full (\FWF) under cyclic analysis (\FWF
evicts all pages upon a miss on a full cache). In the single-core setting, the
advantage of lazy algorithms over \FWF is strict and trivial: for any sequence,
the cost of LRU is no more than \FWF. In the \mc setting, however, such
separation requires careful design and mapping with a bijection on the entire
universe of \inps (\thmref{bijectivefwf}) under \shortnewba.

\vspace*{1mm} \textbf{Separation of \lru via \shortnewba.} Our main contribution
is to show the strict advantage of a variant of \lru over all other lazy
algorithms under \shortnewba combined with a measure of locality
(\thmref{lru-sep}). Although \lru is equivalent to all other lazy algorithms
without restriction on the \inps under \shortnewba, it performs strictly better
in practice~\cite{SilbGaGa14}. This is due to the locality of reference that is
present in real-world inputs~\cite{AlbersFaGi02, Denning68, ChroNo99}. In order
to capture the advantage of \lru, we apply \shortnewba on a universe that is
restricted to \inps with locality of reference~\cite{AlbersFaGi02} and show that
\lru is strictly better than any other lazy algorithm. 

\vspace*{1mm} \textbf{Map.}  The remainder of the paper is organized as follows.
\secref{prelim} describes the model of \mc \paging and provides definitions used
in the rest of the paper. \secref{cba} introduces \fullnewba and establishes
some useful properties of this measure.  \secref{bijective} applies \shortnewba
to establish the advantage of \lm algorithms over non-lazy
\FWF. 
\secref{surjective} 
shows the advantage of \lru over all other lazy algorithms under \shortnewba on
\inps with locality of reference. ~\secref{related} reviews related models of
\mcp, and ~\secref{conclusion} includes a few concluding remarks.  All omitted
proofs can be found in the full version. 

\section{Problem definition}\label{sec:prelim}
This section reviews the free-interleaving model~\cite{LopezSa12a, KattiRa12} of
\mcp and the cost models that are used in this paper. The free-interleaving
model is inspired by real-world architectures and captures the essential
aspects of the \mcp problem.

Assume we are given a \mc processor with $p$ cores labeled
$\core_1, \core_2, \ldots, \core_p$ and a shared cache with $k$ pages
($k \gg p$).

\vspace*{1mm}
\textbf{Input description.}  An \inp to the \defn{\mc \paging problem} is formed
by $p$ online sequences $\req = ( \req_1, \ldots, \req_p )$. Each core
$\core_i$ must serve its corresponding \defn{request sequence}
$\req_i = \langle \pagereq{1}, \ldots, \pagereq{n_i} \rangle$ made up of $n_i$
page requests. The total number of page requests is therefore
$n = \sum_{1\leq i\leq p} n_i$. Given a page (or sequence of pages) $\alpha$ and
a number of repetitions $r$, let $\alpha^r$ denote $r$ repetitions of requests
to $\alpha$.
We assume that for all values of $i$, the length of the request sequence $n_i$
is arbitrarily larger than $k$. That is, we assume that $k \in \Theta(1)$,
which is consistent with the common assumption that parameters like $k$ and
$\tau$ are constant compared to the length of the input.

All requests \pageij are drawn from a finite universe of possible pages
$U$.  Throughout this paper, we assume that request sequences for different
cores may share requests to the same page.  In practice, cores may share their
requests because of races, or concurrent accesses to the same page.

\vspace*{1mm} \textbf{Serving inputs.}  Page requests arrive at discrete
\tss. The requests issued by each core should be served in the same order that
they appear and in an online manner. More precisely, for all $i, j \geq 1$, core
$\core_i$ must serve request \pageij before $\pagereq{j+1}$, and $\pagereq{j+1}$
is not revealed before \pageij is served. The \mc processor may serve at most
$p$ page requests in parallel (up to one request per core\footnote{In practice,
  a single instruction of a core may involve more than one page, but we assume
  that each request is to one page in order to model RISC architectures with
  separate data and instruction caches~\cite{LopezSa12a}.}). Each page request
must be served as soon as it arrives. To serve a request to some page \pageij in
sequence $\req_i$, core $\core_i$ either has a \defn{hit}, when \pageij is
already in the cache, or incurs a \defn{miss}
when \pageij is not present in the cache. In case of a miss, the requested page
should be fetched into the cache. It takes $\tau$ timesteps to fetch a page into
the cache, where $\tau$ is an integer parameter of the problem. During these
timesteps, $\core_i$ cannot see any of its forthcoming requests, that is,
$\pagereq{j+1}$ is not revealed to $\core_i$ before \pageij is \defn{fully
  fetched}.  In case some other core $\core^* \neq \core_i$ is already fetching
the page when the miss occurs, $\core_i$ waits for less than $\tau$ timesteps
until the page is fully fetched to the cache.

\vspace*{1mm} \textbf{Free-interleaving model.}  A \mc \paging algorithm \ma
reads requests from request sequences in parallel and is defined by its eviction
decisions at each timestep. If a core misses while the cache is full, \ma must
evict a page to make space for the requested page before fetching it.  We
continue the convention~\cite{Hass10, LopezSa12a} that when a page is evicted,
the cache cell that previously held the evicted page is unused until the
replacement page is fetched. Finally, the processor serves requests from
different request sequences in the same timestep in some fixed order (\eg by
core index). In today's \mc systems, requests from multiple cores may reach a
shared cache simultaneously. If one core is delayed due to a request to a page
not in the cache, other cores may continue to make
requests. 
~\figref{example-serve} contains an example of serving an \inp with
\textsc{Least-Recently-Used} (LRU)~\cite{SleaTa85, Hass10, LopezSa12a} under
free interleaving.

\vspace*{1mm}
\textbf{Schedule.}  \MC \paging differs from single-core \paging because of the
scheduling component as a result of the fetch delay. The fetch delay slows down
cores at different rates depending on the 
misses they experience, and
requests with the same index on different cores may be served at different times
depending on previous evictions. In other words, the eviction strategy
implicitly defines a \defn{schedule}, or an ordering in which the requested
pages are served by an algorithm. Given an \inp \req defined by $p$ sequences,
the schedule of a \paging algorithm can be 
represented with a copy of \req in which some requests are repeated. These extra
requests captures the \ts at which the processor serves requests from that
input sequence using the \paging algorithm. That is, a schedule has all the same
requests as the corresponding input, but repeats page requests upon a miss
until the page has been fully fetched.

The schedule produced by \lru in the example \inp in~\figref{example-serve} is
the underlined request at each \ts (a formal definition of a schedule can be
found in~\secref{surjective}).

 \textbf{Cost model.} We use the \defn{total time} to measure
algorithm performance and denote the cost that an algorithm \ma incurs on \inp
\req with $\ma(\req)$. The non-competitiveness results from prior work in terms
of the number of misses also hold under the total time~\cite{LopezSa12a, KamaliXu20}.
\begin{figure*}[t]
\begin{center}
\scalebox{.9}{
\begin{tabular}{ c c r c c c }
  Timestep ($t$) & Cache before $t$ & $\req_1$, $\req_2$
  & Status & Schedule $S_{\req, \lru}[t]$ \\
  \hline \hline
  0& $\bot\bot\bot\bot$ & $\underline{a_1}a_2a_1a_5$  & \small{$P_1$ misses, starts fetching $a_1$}
   & $(a_1, a_3)$ \\
   & &  $\underline{a_3}a_4a_5a_2$  & \small{$P_2$ misses, starts fetching $a_3$} & & \\
  \hline
  1 & $\bot\bot\bot\bot$& ${\underline{a_1}}a_2a_1a_5$ & \small{$P_1$ is fetching $a_1$}  & $(a_1, a_3)$
\\
  &   &     ${\underline{a_3}}a_4a_5a_2$& \small{$P_2$ is fetching $a_3$}  & & \\
\hline
  2& $\bot\bot\bot\bot$ & ${\underline{a_1}}a_2a_1a_5$

  & \small{$P_1$ completes fetching $a_1$}  & $(a_1, a_3)$
  \\
    & &  ${\underline{a_3}}a_4a_5a_2$& \small{$P_2$ completes fetching $a_3$ } & & \\
\hline
  3& $a_1a_3\bot\bot$ & ${\color{lightgray}a_1}\underline{a_2}a_1a_5$

  &  \small{$P_1$ misses, starts fetching $a_2$}  & $(a_2, a_4)$
  \\
    & &  ${\color{lightgray}a_3}\underline{a_4}a_5a_2$     & \small{$P_2$ misses, starts fetching $a_4$ } & & \\
\hline
  4 & $a_1a_3\bot\bot$ & ${\color{lightgray}a_1}\underline{a_2}a_1a_5$
  & \small{$P_1$ is fetching $a_2$ } &  $(a_2, a_4)$
  \\
     & & ${\color{lightgray}a_3}\underline{a_4}a_5a_2$   & \small{$P_2$ is fetching $a_4$} & & \\
\hline
  5& $a_1a_3\bot\bot$ & ${\color{lightgray}a_1}\underline{a_2}a_1a_5$
  & \small{$P_1$ completes fetching $a_2$  } &  $(a_2, a_4)$
  \\
     & &  ${\color{lightgray}a_3}\underline{a_4}a_5a_2$   & \small{$P_2$ completes fetching $a_4$ } & & \\
\hline
  6 & $a_1a_3a_2a_4$& ${\color{lightgray}a_1a_2}\underline{a_1}a_5$
  &  \small{$P_1$ has a hit for $a_1$}  & $(a_1, a_5)$
  \\
   &  & ${\color{lightgray}a_3a_4}\underline{a_5}a_2$   & \small{$P_2$ misses, starts fetching $a_5$} & & \\
   & & & \small{($a_3$ is the least-recently-used page and evicted)} & & \\
\hline
  7 &  $a_1\bot a_2a_4$ & ${\color{lightgray}a_1a_2a_1}\underline{a_5}$
    & \small{$P_1$ misses, waits for $a_5$} & $(a_5, a_5)$
  \\
    & & ${\color{lightgray}a_3a_4}\underline{a_5}a_2$     & \small{$P_2$ is fetching $a_5$} & & \\
\hline
  8 & $a_1\bot a_2a_4$ & ${\color{lightgray}a_1a_2a_1}\underline{a_5}$ & \small{$P_1$ completes serving $a_5$}  & $(a_5, a_5)$
  \\
  &   &${\color{lightgray}a_3a_4}\underline{a_5}a_2$    &  \small{$P_2$ completes fetching (and serving) $a_5$} & & \\
\hline
  9& $a_1a_5a_2a_4$ &  ${\color{lightgray}a_1a_2a_1a_5}$  & \small{$P_1$ has completed $\req_1$}   & $(\bot, a_2)$
  \\
  &   & ${\color{lightgray}a_3a_4a_5}\underline{a_2}$   & \small{$P_2$ has a hit for $a_2$, completes $\req_2$}  & & \\
   \hline
\end{tabular}
}
\end{center}
\caption{ Example of execution of \lru on the input $\req = (\req_1, \req_2)$,
  with $\req_1 = \langle a_1a_2a_1a_5 \rangle$ and
  $\req_2 = \langle a_3a_4a_5a_2\rangle$. The cache size is $k = 4$ and the
  fetch delay is $\tau = 3$. We use $\bot$ in the cache to denote an empty slot
  or slot reserved for a page currently being fetched. \newline
  \hspace{\textwidth} If a request incurs a miss, we repeat it in the schedule
  at most $\tau$ times (or however long it takes to be fetched, if some other
  processor already requested it but it has not yet been brought to cache). For
  example, in timestep $7$, we wait two timesteps for $a_5$ to come to the cache
  for $\core_1$ because there were two more steps until $a_5$ was brought to the
  cache by $\core_2$. \newline \hspace{\textwidth}
  In the ``Cache before $t$'' column, we keep track of the state of
  the cache before each timestep.
  The rightmost column is the schedule generated by \lru serving \req. The
  makespan of \req under \lru is $10$. \newline \hspace{\textwidth} The schedule
  for the two cores $\core_1$ and $\core_2$ is defined respectively with
  \protect\ensuremath{\langle a_1, {\color{gray}a_1, a_1},
    a_2,{\color{gray}a_2,a_2},a_1,a_5,{\color{gray}a_5},\bot \rangle} and
  \protect\ensuremath{\langle a_3, {\protect\color{gray}a_3,
      a_3},a_4,{\protect\color{gray}a_4,a_4},a_5,\protect{\color{gray}a_5,a_5},a_2
    \rangle}.}
\label{fig:example-serve}
\end{figure*}


\begin{definition}[Total time]
  The total time an algorithm \ma takes to serve an \inp $\req$ is the sum of
  the timesteps it takes for all cores to serve their respective request
  sequences.  That is, the total time
  $\ma(\req) = \sum\limits_{1 \leq i \leq p} \ma(\req_i)$ where $\ma(\req_i)$
  denotes the timesteps $\core_i$ took to serve $\req_i$ with algorithm \ma.
\end{definition}

Total time combines aspects from both makespan and the number of misses,
the two cost measures in previous studies of \mc \paging~\cite{LopezSa12a,
  KattiRa12}. The makespan is the maximum time it takes any core to complete its
request sequence, and hence is bounded above by total time.  Specifically, the
total time is monotonically increasing with respect to both the number of misses
and the makespan.

The total time is a more realistic measure of performance than the number of
misses because it determines performance in terms of the time that it takes to
serve the input. In contrast, the number of misses does not directly correspond
with the time to serve an \inp because a miss may take less than $\tau$ steps
to fetch the page if it is already in the process of being fetched by another
core. 
%
The total time also captures aspects of algorithm performance that are not
addressed by makespan.  In particular, makespan does not capture the overall
performance of all cores.  For example, a solution in which all cores complete
at \ts $t$ has a better makespan than a solution in which one core completes at
\ts $t+1$ while the rest complete much earlier, e.g.\ at \ts $t/2$. The second
solution is preferred in practice (and also under the total time) as most cores
are freed up earlier.

%

\section{\fullnewbacap for online problems}\label{sec:cba}

We define a new analysis measure called \defn{\fullnewba} inspired by bijective
analysis~\cite{AngeReSc20,AngeDoLo07,AngeSc09, AngeDoLo08,Dorr10} and explore
alternative paths to showing relationships between algorithms under \shortnewba
via a relaxed measure called ``natural surjective'' analysis.
\fullnewbacap extends the
advantages of bijective analysis to online problems with multiple input
sequences. 

\vspace*{1mm} \textbf{Overview.} Although traditional bijective analysis has
been applied to compare single-core \paging algorithms, it requires modification
to capture the notion of ``input length" in \mcp. Since each request sequence in
an \inp for \mcp may have a different length in terms of the number of requests,
there are multiple ways to define the length of an \inp. It is not clear which
definition of length is most natural or correct for \mcp.

Furthermore, partitioning the input space based on the number of requests in an
\inp as in bijective analysis for single-core paging may be overly restrictive
for \mcp, because the time it takes to serve \inps of the same length (in terms
of the number of requests) may differ depending on the algorithm. In \mcp, the
time depends on the interleaving of the multiple request sequences. Cyclic
analysis addresses these issues by removing the restriction that bijections
should be drawn between \inps of the same length.

At a high level, in order to show a relationship between two algorithms \ma and
\mb under bijective analysis or \shortnewba, one must define a mapping between
inputs and their costs under different algorithms. One way to model mappings
between inputs with different costs 
is with a \defn{\cgraph}. Given algorithms $\ma$ and $\mb$, an \cgraph is an
infinite directed graph where the nodes represent inputs and there exists an
edge from \inp $\req_1$ to \inp $\req_2$ if and only if
$\ma(\req_1) \leq \mb(\req_2)$. In order to show the advantage of algorithm
$\ma$ over $\mb$, traditional bijective analysis partitions the (infinite) graph
of inputs into finite subgraphs, each formed by inputs of the same
length. Within each partition, the bijection relating $\ma$ to $\mb$ defines a
set of cycles such that each vertex is in exactly one cycle of finite length
{(cycles may have length one, i.e.\ they may be self-loops).}
Cyclic analysis relaxes the requirement that all subgraphs in the partition must
be finite, but also requires that each node in each induced subgraph must have
an in-degree and out-degree of one. That is, each node in the induced subgraph
is part of a cycle.

\vspace*{1mm} \textbf{Measure definition and discussion.}  Let \univ denote the
(infinite) set of all \inps, and for an algorithm \ma and \inp $\req \in \univ$,
let $\ma(\req)$ denote the cost \ma incurs while serving \req. The notation in
our discussions of \fullnewba is inspired by ~\cite{AngeDoLo07}.

\begin{definition}[\fullnewbacap]
    \label{def:bij-analysis}
    We say that an online algorithm \ma is \defn{no worse} than online algorithm
    \mb under \defn{\fullnewba} if there exists a
    bijection $\pi : \univ \leftrightarrow \univ$ satisfying
    $\ma(\req) \leq \mb(\pi(\req))$ for each $\req \in \univ$. We denote this by
    $\ma \bleq \mb$. Otherwise we denote the situation by $\ma \bnleq
    \mb$. Similarly, we say that \ma and \mb are the same according to
    \shortnewba if $\ma \bleq \mb$ and $\mb \bleq \ma$. This is denoted by
    $\ma \bequiv \mb$. Finally we say $\ma$ is better than \mb according to
    \shortnewba if $\ma \bleq \mb$ and $\mb \bnleq \ma$. We denote this by
    $\ma \bless \mb$.
  \end{definition}

  Bijective analysis is defined similarly, except that the input universe is
  partitioned based on the length of inputs, and bijections need to be drawn
  between inputs inside each partition. In contrast, cyclic analysis allows
  mapping arbitrary sequences to each other. Bijective analysis and \shortnewba
  have several benefits over competitive analysis~\cite{AngeSc09}. Specifically,
  they:
  \begin{itemize}[leftmargin=*]
    \setlength\itemsep{-.25em}
  \item \emph{capture overall performance}.  If $\ma \bleq \mb$, every ``bad''
    \inp for algorithm $\ma$ corresponds to another \inp for algorithm \mb which
    is at least as bad. Hence, the performance of algorithms is evaluated over
    all request sequences rather than a single worst-case sequence.
  \item \emph{avoid comparing to an offline algorithm}.  Competitive analysis is
    inherently pessimistic as it compares online algorithms based on their
    worst-case performance against a powerful adversary. This pessimism is
    especially pronounced in multicore \paging where an offline algorithm can
    ``artificially'' miss on some pages in order to schedule sequences in a way
    to minimize its total cost. This scheduling power is a great advantage for
    \opt as shown in~\cite{LopezSa12a}. Instead, we use \shortnewba because
  compares online algorithms directly without involving an offline algorithm.
\item \emph{can incorporate assumptions about the universe of
    \inps}. \fullnewbacap can also define relationships between algorithm
  performance on a subset $S \subset \univ$ of \inps. For example, applying
  \shortnewba to a restricted universe of \inps with locality of reference has
  been used to separate \lru from other algorithms in the single-core
  setting~\cite{AngeDoLo07,AngeSc09}. Since \lru exploits locality of reference,
  analyzing \inps with locality may yield a better understanding of the
  performance of algorithms. Most other measures such as competitive ratio are
  unable to separate \lru from other lazy algorithms~\cite{AngeDoLo07}.
\end{itemize}

As mentioned above, bijective analysis, as defined for single-core \paging
~\cite{AngeDoLo07,AngeSc09}, requires partitioning the universe of \inps $\univ$
into finite sets of inputs of the same length. For multicore \paging, however,
this partitioning 
is not necessary nor well-defined. In fact, for many online problems, the length
of input is not necessarily a measure of ``difficulty'', as trivial request
(e.g., repeating requests to a page) can artificially increase the
length. 
As
such, there is no priory reason to draw bijections between sequences of the same
length.

For problems such as single-core \paging and list update~\cite{AngeSc09}, where the input is
formed by a single sequence, the length of the input is simply the length of the
sequence. In \mc \paging, however, the length of \inps is not
well-defined as multiple sequences are involved.  
Should the length be
the sum of the number of requests or a vector of lengths for each request
sequence?  To address these issues, \shortnewba generalizes the finite
partitions of \ba to the entire universe of \inps. This would give \shortnewba a
flexibility that makes it possible to study other problems under this measure.
We note that, the restrictive nature of bijective analysis not only makes it
hard to study algorithms under this measure, but also can cause situations that
many algorithms are not comparable at all. The following example illustrates the
restriction of \ba when compared to \shortnewba:

\vspace*{1mm}
\textbf{Example.} Consider two algorithms \ma and \mb for an online problem P
(with a single sequence as its input).  Assume the costs of \ma and \mb are the
same over all inputs, except for four sequences.  Among these four, suppose that
two sequences $\sigma_1$ and $\sigma_2$ have the same length $m$ and we have
$\ma(\sigma_1) = 10$ and $\ma(\sigma_2) =40$ while $\mb(\sigma_1) = 20$ and
$\mb(\sigma_2)=30$. For inputs of length $m$, there is no way to define a
bijection that shows advantage of one algorithm over another. So, the two
algorithms are incomparable under bijective analysis. Next, assume for sequences
$\sigma_3$ and $\sigma_4$ we have
$\ma(\sigma_3)=20, \ma(\sigma_4)=30, \mb(\sigma_3)=40$, and
$\mb(\sigma_4)=20$. The following mappings shows $\ma \bless \mb$:
$\sigma_1 \rightarrow \sigma_1, \sigma_2\rightarrow \sigma_3,
\sigma_3\rightarrow \sigma_4$, and $\sigma_4\rightarrow \sigma_2$.

\vspace*{1mm}
\textbf{Bounding \inps with the same cost.}  In order for \fullnewba to be a
meaningful measure, 
there must
not be an infinite number of \inps that achieve the same cost. To be more precise, for the
universe of inputs \univ and an algorithm \ma, let $\ma(\univ)$ be the
corresponding multiset of costs associated with inputs in \univ.

\begin{definition}[\Ccp]
  A cost measure for an online problem satisfies the \ccp if and
  only if for any algorithm $\ma$ and for all unique costs $m \in \ma(\univ)$,
  the set of inputs that achieve that cost is bounded.
\end{definition}

If a cost measure does not satisfy the \ccp, it is possible to
prove contradicting results under \shortnewba. That is, if there are infinitely
many \inps that achieve each cost, for any algorithms $\ma, \mb$, it is possible
to define bijections such that $\ma \bless \mb$ and $\mb \bless \ma$.

In the case of \mc \paging, the total time and makespan cost models both have the
\ccp while the miss count and the closely related miss rate do not. For
example, the infinitely many sequences that only request some page $\alpha$
(e.g. $\alpha$, $\alpha\alpha$, $\alpha\alpha\alpha$, \ldots) all have cost one
under miss count, but all have different costs under total time and makespan.


The following lemma guarantees that \fullnewba has the ``to-be-expected''
property that if algorithm $\ma$ is better than $\mb$, $\mb$ is not better than
$\ma$. In the case of bijective analysis, this property easily follows from the
fact that bijections are drawn in finite sets (formed by inputs of the same
length). Since the bijections in \fullnewba are defined in an infinite space, a
more careful analysis is required.

  \begin{restatable}{lemma}{ccpcontra}
  Given algorithms $\ma, \mb$ for a problem satisfying the \ccp, it is
  not possible that $\ma \bless \mb$ and $\mb \bless \ma$ at the same time.
  \end{restatable}

\begin{proof}
  If $\ma \bless \mb$, by~\defref{bij-analysis}, there must exist an \inp
  $\sigma \in \univ$ such that $\ma(\sigma) < \mb(\pi(\sigma))$. Let $\sigma$ be
  the \inp with the smallest cost that differs between $\ma, \mb$, and let
  $\univ^{\ma(\sigma)}_\ma, \univ^{\ma(\sigma)}_\mb \subset \univ$ be the
  sequences that have cost at most $\ma(\sigma)$ in $\ma(\univ), \mb(\univ)$,
  respectively. By the \ccp, $ |\univ^{\ma(\sigma)}_\ma|$ and
  $|\univ^{\ma(\sigma)}_\mb|$ are both bounded and
  $|\univ^{\ma(\sigma)}_\ma| > | \univ^{\ma(\sigma)}_\mb|$. It is impossible to
  define another function $\phi$ such that $\mb \bleq \ma$ because there are not
  enough \inps in $\univ^{\ma(\sigma)}_\mb$ to map to all \inps in
  $\univ^{\ma(\sigma)}_\ma$ such that the cost of each \inp under $\mb$ is at
  most the cost of the corresponding \inp under $\ma$.
\end{proof}

  Similarly, if $\ma \bless \mb$, then $\ma \not\equiv_c \mb$ for problems with the
  \ccp. Additionally, \shortnewba has the transitive property: if
  $\ma \bleq \mb$ and $\mb \bleq \mathcal{C}$, then $\ma \bleq \mathcal{C}$.
  The \ccp guarantees that each node in the \cgraph has infinite out-degree but
  finite in-degree because each input has infinitely many inputs that cost more
  than it and finitely many inputs that cost less than it.

\vspace*{1mm}
\textbf{Relation of surjectivity to \shortnewba.}
In the remainder of the section we will discuss the role of surjective mappings
as an intermediate step before defining a bijective mapping between infinite
sets. In traditional bijective analysis, since the \inp set is finite because of
the length restriction, any surjective mapping must also be bijective.  In some
problems, including \mc \paging, it may be easier to define a surjective mapping
between the \inps. We will first show that a class of surjective mappings can be
converted into bijective mappings.

Suppose we have a surjective but not necessarily injective mapping between two
infinite sets $f: X \rightarrow Y$.  For all positive integers
$m \in \mathbb{N}$, let $X_m \subseteq X, Y_m \subseteq Y$ be subsets of the
pre-image and image respectively such that exactly $m$ elements in $X_m$ map to
one element in $Y_m$. That is, given some $m$, $x \in X_m$ implies that there
are $m-1$ other elements $x_1, x_2, \ldots, x_{m-1} \neq x$ such that for
$i = 1, \ldots, m-1$, $f(x) = f(x_i)$. Each $X_m, Y_m$ is an element of a
partition of the pre-image and image, respectively.

\begin{definition}[\surjcap surjective mapping]
  \label{def:nat-surj}
  Given a surjective function $f: X \rightarrow Y$, $f$ is \defn{\surj} if and
  only if for all $m \in \mathbb{N}$, the partitions $X_m$ and $Y_m$ are either
  empty or infinite.
\end{definition}

For example, the function
$f: \mathbb{N} \rightarrow \mathbb{N}, f(x) = \floor{x/2}$ is a \surj surjective
mapping (assuming $0 \in \mathbb{N}$) because exactly two elements in the
pre-image map to each element in the image. In contrast,
$g: \mathbb{Z} \rightarrow \mathbb{N}, g(x) = |x|$ is not natural because there
is only one element in $X_1$ and $Y_1$ at $x = y = 0$.

We introduce \defn{natural surjective (NS) analysis}, a technique to compare
algorithms under \shortnewba using an intermediate surjective \emph{but not
  injective} mapping. The formalization is almost identical
to~\defref{bij-analysis}, but the function $\pi$ needs only to be a natural
surjective function. We use $\sleq$ to denote the relation between two
algorithms under NS analysis.  In the rest of the paper, we will refer to
natural surjective functions and natural surjective analysis as surjective
functions and surjective analysis, respectively.


\begin{figure}[t]
  \centering
  \scalebox{.8}{
    \begin{tikzpicture}[shorten >=1pt,node distance=1.0cm,on grid,auto] 
     \Large
        	\node at (0,3.75) {$X_2$};            	\node at (1.85,3.75) {$Y_2$};
        	\node at (6,3.75) {$X_2$};            	\node at (7.85,3.75) {$Y_2$};
    \large
	\node at (0,3) {$x_{1,1}$}; \draw[->] (0.5,3) -- (1.5,3);  \node at (1.9,3) {$y_1$};
	\node at (0,2.25) {$x_{1,2}$}; \draw[->] (0.5,2.25) -- (1.5,2.8);
	\node at (0,1.5) {$x_{2,1}$}; \draw[->] (0.5,1.5) -- (1.5,1.5);  \node at (1.9,1.5) {$y_2$};
    \node at (0,.75) {$x_{2,2}$}; \draw[->] (0.5,0.75) -- (1.5,1.3);
    \node at (0,0) {$\ldots$};      \node at (1.75,0) {$\ldots$};
    \Large

    \node at (3.6,1.75)  {$\xRightarrow[\text{\ }]{\text{unzip}}$};

    \large
	\node at (6,3) {$x_{1,1}$}; \draw[->] (6.5,3) -- (7.5,3);  \node at (7.9,3) {$y_1$};
	\node at (6,2.25) {$x_{1,2}$}; \draw[->] (6.5,2.25) -- (7.5,2.25);      \node at (7.9,2.25) {$y_2$};
	\node at (6,1.5) {$x_{2,1}$}; \draw[->] (6.5,1.5) -- (7.5,1.5);  \node at (7.9,1.5) {$y_3$};
    \node at (6,.75) {$x_{2,2}$}; \draw[->] (6.5,0.75) -- (7.5,0.75);\node at (7.9,.75) {$y_4$};
    \node at (6,0) {$\ldots$};      \node at (7.75,0) {$\ldots$};
  \end{tikzpicture}
  }
  \caption{ Example of unzipping $X_2, Y_2$ in a \surj surjective mapping.}
 \label{fig:unzip}
\end{figure}
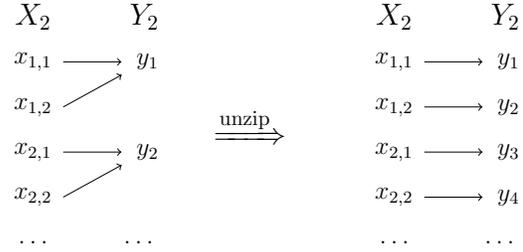

\begin{restatable}[``Unzipping'' equivalence]{lemma}{unzipping}
  \label{lem:surj-bij-equiv}
  Let algorithms $\ma, \mb$ be algorithms for a problem with the \ccp. If
  $\ma \sleq \mb$ under a \surj surjective mapping, then $\ma \bless \mb$.
  \end{restatable}

\begin{proof}
  At a high level, we will describe how to convert a \surj surjective function
  $f$ into a bijective mapping $f_b$ by ``unzipping'' any many-to-one mappings
  in each partition. At a high level, the new mapping $f_b$  ``remaps'' elements
  in the preimage to elements in the image.

  Let $X_m, Y_m$ be the pre-image and image of a non-empty mapping-based
  partition for any fixed $m \in \mathbb{N}$. Suppose we order the elements in
  $Y_m$ from lowest to highest and let $y_i$ be the $i$-th largest element in
  $Y_m$. The elements in any set $Y_m$ can be ordered because of the \ccp. Given
  an element $y^m_i \in Y_m$, let the corresponding elements in the pre-image be
  $x_{i, j} \in X_m$ for $j = 1, 2, \ldots, m$ in some order. Since
  $\ma(x_{i, j}) \leq \mb(y^m_i)$ for all $i, j$ (by the definition of surjective
  analysis), for any $i, j$, $\ma(x_{i, j}) \leq \mb(y^m_z)$ for $z >
  i$. Therefore, we define a new bijective mapping $f_b$ based on $f$ such that
  $f_b(x_{i, j}) = y^m_{mi + j - 1}$. The new mapping $f_b$ satisfies the property
  that for all $\sigma \in \univ$, $\ma(\sigma) \leq \mb(f_b(\sigma))$.
\end{proof}

 As shown in the example in~\figref{unzip}, we can convert a natural surjective
  mapping to a bijective one by ``unzipping'' the mapping and maintaining the
  relative order of \inps. 




The relationship between surjective analysis and \shortnewba allows
for different paths to proving relationships between algorithms. In traditional
bijective analysis, we had to define a direct bijection between two algorithms
because all surjections are bijections in finite sets of the same size. Natural
surjective analysis is a potentially easier proof technique that is equivalent to 
\shortnewba.


\section{\fullnewbacap for \mcp}\label{sec:bijective}
It is straightforward to show that all \lm \mcp algorithms are equivalent under
cyclic analysis (see~\proporef{lazybijective} for a full proof). Therefore, to
show a separation between two algorithms, we analyze a variant of \FWF that
flushes (empties) the entire cache if it incurs a miss when the cache is full.
In what follows, we show the advantage of \lm algorithms over \FWF.  While this
result is not surprising, the techniques used in the proofs prepare the reader
for the more complicated proof in the next section.
%

\begin{restatable}{lemma}{nicebijection}
  \label{lem:niceBijection}
  Assume $p=2$. Consider two \lm \paging algorithms \ma and \mb which have the
  same eviction policy starting at the same timestep $t$ and cache contents at
  $t$ except for one page $x$ that is present in the cache of \ma and absent in
  the cache of \mb. If \ma and \mb incurred the same cost up until timestep $t$,
  we have $\ma \bless \mb$.
\end{restatable}

\begin{proof}
  At a high level, we will define a surjective cyclic mapping on the input space
  with cycles of length 2.
    For \inps where $x$ is never requested before being evicted, \ma and \mb
  perform similarly. We assume these \inps are mapped to themselves and ignore
  them (the cycles associated with these inputs are self-loops). In the
  remainder of the proof, we assume $x$ is requested at timestep $t$ before
  being evicted. At timestep $t$, \ma has a hit on the request to $x$ while
  \mb incurs a miss. As a result, the schedule of the two algorithms (i.e., the
  order at which they serve the requests) becomes different after serving $x$
  and hence there is no guarantee that \ma has less cost that
  \mb.

  We define a bijection $b$ in a way that the schedule of \ma for any input $R$
  is similar to that of \mb for serving $b(\req)$. The bijection that we define
  creates cycles of length 2: if $\req' = b(\req)$ then $\req =
  b(\req')$; 
  we denote this by $\req \leftrightarrow \req'$.

  Let $\core_1$ and $\core_2$ denote the two cores and let
  \cont{$\sigma_1$}{$\sigma_2$} denote the \defn{continuation} of a sequence
  where $\core_1$ asks for sequence $\sigma_1$ and $\core_2$ asks for $\sigma_2$
  from time $t$ onward.
We define the bijection based on two cases. In both cases, one of the cores, say
$\core_2$, has a request to page $x$ at time $t$ and hence \ma and \mb perform
differently on the continuation of the sequence. Assume the contents of the
caches of \ma and \mb at time $t$ are respectively $H \cup \{x\}$ and $H$.

\hspace{1mm} \textbf{Case 1:} $\core_1$ requests a page $q \notin H$.

Recall that $\core_2$ asks for $x$ at time $t$, so the input can be written as
$\req = $ \cont{$q \sigma$}{$x \sigma'$} for some $\sigma$ and $\sigma'$.  We
define $\req' = $ \cont{$x^\tau \sigma$}{$q \sigma'$}.
To show the mapping $\req \leftrightarrow \req'$ is a valid mapping we need to
show $\ma(\req) \leq \mb(\req')$ and $\ma(\req') \leq \mb(\req)$. First, we show
$\ma(\req) \leq \mb(\req')$.
%
%
On input \req, \ma has a miss on $q$ and a hit on $x$ at time $t$; so, \ma
starts serving $\sigma$ and $\sigma'$ at \tss $t+\tau$ and $t+1$, respectively,
it serves $\sigma$ exactly $\tau-1$ \tss later than $\sigma'$. On input $\req'$,
\mb has a miss on both $x$ and $q$ at time $t$. It incurs an additional
$\tau-1$ hits on $x$ after fetching it. So, \mb starts serving $\sigma$ and
$\sigma'$ at \tss $t+\tau + (\tau-1)$ and $t+\tau$, respectively. In other
words, it serves $\sigma$ exactly $\tau-1$ \tss later than $\sigma'$. The
content of the cache of \ma and \mb is the same for serving $\sigma$ and
$\sigma'$. We conclude that the number of misses (and hence total time) of \mb
in serving $\sigma$ and $\sigma'$ in $R$ is the same as \ma in $\req'$. For the
first requests to $q$ and $x$ in $\req$, \ma incurs one miss (and total time
$\tau+1$) while \mb incurs two misses (and total time $3\tau-1$) for the first
requests to $x^\tau$ and $q$ in $\req'$. We conclude that $\ma(\req) <
\mb(\req')$.
%
To complete the proof in Case 1, we should show $\ma(\req') \leq
\mb(\req)$. When \ma serves $\req'$, it incurs $\tau$ hits on $x^\tau$ and one
miss on $q$; as such, it starts serving $\sigma$ and $\sigma'$ at the same time
$t+\tau$. On the other hand, when \mb serves $\req$, it incurs a miss on both
$q$ and $x$ and starts serving $\sigma$ and $\sigma'$ at the same time
$t+\tau$. So, the two algorithms incur the same cost for serving $\sigma$ and
$\sigma'$. Moreover, $\ma$ one miss and $\tau$ hits (and total time $2\tau$)
for serving $x^\tau$ and $q$ while \mb incurs two misses (and total time
$2\tau$) for serving $q$ and $x$, so $\ma(\req') = \mb(\req)$.


  \hspace{1mm} \textbf{Case 2:} $\core_1$ asks for a page $a \in H$.

  So, the input can be written as $\req = $ \contt{$a \sigma$}{$x \sigma'$} for
  some sequence of requests $\sigma$ and $\sigma'$. We define $\req' = $
  \contt{$x \sigma$}{$a^\tau \sigma'$}. To show the mapping
  $\req\leftrightarrow \req'$ is a valid mapping we first show
  $\ma(\req) \leq \mb(\req')$. \ma starts serving both $\sigma$ and $\sigma$ in
  $\req$ at $t+ 1$ because \ma has hits on both $a$ and $x$. On the other hand,
  \mb has a miss on $x$ and a hit on $a$ when serving all copies of $\tau$. That
  means, it starts serving both $\sigma$ and $\sigma'$ in $\req'$ at the same
  time $t+\tau$. The content of the cache of the two algorithms is also the same
  ($x$ is now in the cache of \mb).
  So, \ma and \mb incur the same number of misses (and total time) for both
  $\sigma$ and $\sigma'$. For the prefixes $a$ and $x$ in $\req$, \ma incurs 0
  misses (and total time 2); for the prefixes $a^\tau$ and $x$ in $\req'$, \mb
  incurs 1 miss (and total time $2\tau$). We conclude $\ma(\req)<\mb(\req')$.
  Next, we show $\ma(\req') \leq \mb(\req)$. \ma has hits on all requests in
  $a^\tau$ and $x$ in $\req'$, i.e., it serves $\sigma$ and $\sigma'$ at \tss
  $t+1$ and $t+\tau$, respectively. That is, it serves $\sigma$ exactly $\tau-1$
  units later than $\sigma'$. \mb, on the other hand, has a hit at $a$ and a
  miss at $x$ in $\req$, i.e. it serves $\sigma$ and $\sigma'$ at times $t+1$
  and $t+\tau$, respectively. So, the two algorithms incur the same cost for
  $\sigma$ and $\sigma'$. For the prefixes $a^\tau$ and $x$, \ma incurs 0 misses
  and total time $\tau+1$. For the prefixes $a$ and $x$, \mb incurs 1 miss and
  total time $\tau+1$. We conclude that $\ma(\req') \leq \mb(\req)$.
\end{proof}

\vspace*{2mm}
We show the advantage any lazy algorithm \ma over non-lazy \FWF by comparing
their cache contents at each timestep.

\begin{restatable}{theorem}{bijectivefwf}\label{thm:bijectivefwf}
  Any lazy algorithm \ma is strictly better than \FWF under cyclic analysis for
  $p=2$, that is, $\ma \bless \FWF$.
\end{restatable}

\begin{proof}
  Let $\FWF_i$ be a variant of \FWF which, instead of flushing the cache, evicts
  $i$ pages from the cache; these $i$ pages are selected according to \ma's
  eviction policy. That is, the algorithm evicts $i$ pages that \ma evicts when
  its cache is full (as an example, if \ma is \LRU, the algorithm evicts the $i$
  least-recently-used pages). We will show $\ma \bless \FWF$ by transitivity of
  bijection. In particular, we show
  $$\ma = \FWF_1 \bless \ldots \bless \FWF_{k-1} \bless \FWF_k =
  \FWF.$$

  Let $\FWF_i^t$ be an algorithm that applies $\FWF_i$ for the first $t$
  timesteps and $\FWF_{i+1}$ for timesteps after and including $t+1$. If we can
  show $\FWF_i^{t+1} \bless \FWF_i^t$ for all $t$, again by transitivity of
  bijection, we get $\FWF_i \bless \FWF_{i+1}$. We note that $\FWF_i^{t+1}$ and
  $\FWF_i^{t}$ differ in serving at most one request at time $t$, and they have
  the same eviction strategy for the remainder of the input. If the cores do not
  incur a miss at time $t$, both algorithms perform similarly. For sequences for
  which there is a miss at time $t$, there will be one less page in the cache of
  $\FWF_i^t$ compared to $\FWF_i^{t+1}$. Therefore,
  $\FWF_i^{t+1} \bless \FWF_i^t$ by~\lemref{niceBijection}.
\end{proof}

As the bijection in the proofs illustrates, the main insight of \shortnewba is
the direct comparison of algorithms by drawing mappings between \inps of
different lengths. In contrast 
to the single-core setting, \inps of the same length (in the number of requests)
in \mcp may take different amounts of time, so we define bijections based on the
schedule (and therefore length in time) rather than the number of requests. In
the next section, we use the same idea of mapping sequences with different
lengths in requests but similar schedules.


\section{Advantage of \lru with locality of reference}\label{sec:surjective}

To demonstrate how to use \shortnewba to separate algorithms, this section
sketches the separation of \lru from all other \lm algorithms on \inps with
locality of reference via \shortnewba. Along the way, we demonstrate how to use
surjective analysis to establish relations between algorithms under
\shortnewba. In practice, \lru (and its variants) are empirically better than
all other known \paging algorithms~\cite{SilbGaGa14} because sequences often have
temporal locality.

The full proofs for this section can be found in the full version.

\subsection{Preliminaries}

First, we will formalize the notion of a schedule from~\secref{prelim}, which
represents an algorithm's eviction decisions by repeating requests in an \inp on
a miss. We will use the schedule to later define locality of
reference. Throughout this section, let \ma be a \paging algorithm and \req be an
\inp.

\begin{definition}[Schedule]
  \label{def:schedule-formal}
  The \defn{schedule}
  $\skedar = \{\sked_{\req_1, \ma}, \ldots, \sked_{\req_p, \ma}\}$ is another
  \inp where each request sequence is defined as the implicit schedule that \ma
  generated while serving \req. That is, $\sked_{\req_i,\ma}[t]$ is the request
  that core $\core_i$ serves at timestep $t$ under \ma. Also, \skedar is the
  same as $\req$ with each miss repeated at most $\tau - 1$ times (as many
  repetitions as it takes to resolve the given miss, which might be less than
  $\tau - 1$ if the page was already in the process of being fetched). We use
  $\sked_{\req_i, \ma}[t_1, t_2]$ (for all $i$) to denote all requests
  (including repetitions due to misses) made by $\core_i$ between timesteps
  $t_1$ and $t_2$ (inclusive).
\end{definition}

We use the formal definition of schedule to discuss dividing up an \inp under
\ma based on its schedule up until some \ts.

\begin{definition}[Schedule prefix and suffix]
 \label{def:sked-pre-suf}
 Let \nar be the time required for \ma to serve \req. Given an integer \ts
 $j < \nar$, we define parts of the schedule that will be served before, after,
 and during \ts $j+1$.

 Informally, the \defn{schedule prefix} $\skedprej$ is all the requests served
 up to \ts $j$ with repetitions matching scheduling delay, the schedule at \ts
 $j+1$, $\skedar[j+1]$, is all requests served at \ts $j+1$, and the
 \defn{schedule suffix} \skedsufj is all requests served after \ts $j+1$ with
 repetitions matching scheduling delay.  Note that \skedpre or \skedsuf may be
 empty.  When the \ts $j$ and/or algorithm $\alg$ are clear from context,
 we will drop them from the schedule notation.
\end{definition}

\begin{definition}[Request prefix and suffix]
  \label{def:req-pre-suf}
  Let $\req^{\leq j, \ma}$ be all subsequences from \req served up to \ts $j$,
  $\req^{> j, \ma}$ be all subsequences from \req served after \ts $j$, and
  $r_{j+1}$ be the requests at \ts $j+1$.  For simplicity, we define the
  \defn{request prefix} as $\reqpre = \req^{\leq j, \ma}$ and \defn{request
    suffix} as $\reqsuf = \req^{> j+1, \ma}$ when $j, \ma$ are understood from
  context.
\end{definition}

The request prefix and suffix formalizes the analysis technique
from~\secref{bijective} of defining mappings based on the continuation of the
\inp after some \ts.

Using the \lru example in~\figref{example-serve} when $j = 4$,
$\reqpre_1 = a_1a_2$, $\reqpre_2 = a_3a_4$ because those are the pages that have
been requested until \ts $4$. Since at \ts $5$ all cores are fetching requests,
$r_{j+1} = \emptyset$. Also, $\reqsuf_1 = a_1a_5$ and $\reqsuf_2 = a_5a_2$
because those are the requests remaining after \ts $5$.  Similarly,
$\sked^{\text{pre}}_{\req_1} = a_1a_1a_1a_2$ and
$\sked^{\text{pre}}_{\req_2} = a_3a_3a_3a_4$. At \ts $4$, both cores are
fetching, so $r_{j+1} = (a_2, a_4)$. The suffix is the schedule for \tss after
$4$, so $\sked^{\text{suf}}_{\req_1} =a_2a_1a_5a_5$ and
$\sked^{\text{suf}}_{\req_2} =a_4a_5a_5a_5a_2$.

\paragraph{Locality of reference and the Max-Model.}
We will restrict the space of all inputs with the ``Max-Model'', an
experimentally-validated model of locality of reference that limits the number
of distinct pages in subsequences of an \inp with a concave
function~\cite{AlbersFaGi02}.

We define a \defn{window} of size $w$ in the \mc setting as $p$ runs of
consecutive requests of length $w$ (one for each core). The Max-Model for
\mc \paging is the same as in single-core \paging except that it considers
windows over all cores.

In the Max-Model for \mc \paging, an input $\req$ is \defn{consistent} with some
increasing concave function $f$ if the number of distinct pages in any window of
size $w$ is at most $f(w)$, for any $w \in \mathbb{N}$~\cite{AlbersFaGi02}. That
is, a function $f:\mathbb{N} \rightarrow \mathbb{R}^+$ is \defn{concave} if
$f(1) = p$, and $\forall n \in \mathbb{N}: f(n+1) - f(n) \leq f(n+2) - f(n+1)$.
In the Max-Model, we also require that $f$ is surjective on the integers between
$p$ and its maximum value.

It is easy to adapt \shortnewba to the Max-Model by restricting to \inps
consistent with a concave function $f$ (denoted by $\localset$).  Let
$\alg \bleq^f \mb$ denote that $\alg$ is no worse than $\mb$ on \localset under
\shortnewba.  Similarly, let $\ma \sleq^f \mb$ denote that $\ma$ is no worse
than $\mb$ on \localset under surjective analysis.

\subsection{Advantage of \lru on \inps with locality}

In the rest of the section, we will show that \lru is no worse than sequences
with locality under \shortnewba by establishing a surjective mapping
(\defref{nat-surj}) and converting it into a bijective mapping
(\lemref{surj-bij-equiv}). The main technical challenge in the proof of the
separation of \lru is that sequences with the same number of requests may have
different schedules and therefore may differ significantly in their cost, even
if they only differ in one request. We use \shortnewba to avoid the restriction
of comparing \inps of the same length and instead define a function to relate
\inps of the same cost.

Along the way, we demonstrate how to use surjective analysis as a proof
technique for comparing algorithms via \shortnewba on the entire space of \inps
as described in~\secref{cba}. The construction of the surjective mapping is
inspired by a similar argument in the single-core setting by Angelopoulos and
Schweitzer~\cite{AngeSc09} which establishes a bijective mapping within finite
partitions, but requires a more complex mapping based on schedules.

We will show that for every algorithm \alg, $\lru \sleq^f \alg$. An arbitrary
algorithm $\alg$ may be very different from \lru.  Therefore, instead of
defining a direct bijection, we will use intermediate algorithms
$\mb_1, \ldots, \mb_\ell$ such that
$\alg \equiv \mb_1 \succeq^f_s \ldots \succeq^f_s \mb_i \succeq^f_s \ldots
\succeq^f_s \mb_\ell \equiv \lru$. The result follows from the transitivity of
the ``$\sleq^f$'' relation.  Intuitively, we construct algorithms ``closer'' to
\lru at each step in the series as we will explain in~\lemref{lru-surjection}.
We formalize the notion of an algorithm \ma's ``closeness'' to \lru in terms of
the evictions that it makes.  An algorithm \ma is \defn{\lru -like} at timestep
$t$ if after serving all requests up to time $t-1$, it serves all requests at
time $t$ as \lru would.



\vspace*{1mm}
\paragraph{Defining a surjective  mapping between \inps.}
At a high level, the proof proceeds by defining a surjection between similar
sequences with two pages swapped. We define a ``complement'' of a sequence as a
new sequence with certain pages swapped, and show properties of complements of
sequences with locality required for our main proof.

\begin{definition}[Complement \cite{AngeSc09}] \label{def:complement}
  Let $\beta, \delta$ denote two distinct pages in $U$, the universe of
  pages. Let $\req_i[j]$ denote the $j$-th request in the \ith request sequence
  of an \inp \req. The \defn{complement} of $\req_i[j]$ with respect to $\beta$
  and $\delta$, denoted by $\overline{\req_i[j]}^{(\beta, \delta)}$, is the
  function that replaces $\beta$ with $\delta$, and vice versa. Formally,
  $\overline{\req_i[j]}^{(\beta, \delta)} = \delta$, if $\req_i[j] = \beta$;
  $\overline{\req_i[j]}^{(\beta, \delta)} = \beta$, if $\req_i[j] = \delta$; and
  $\overline{\req_i[j]}^{(\beta, \delta)} = \req_i(j)$,
  otherwise.  \end{definition}

We use $\overline{\req_i[j]}$ when
$\beta, \delta$ are clear from context.  We denote each request sequence
$\req_i = \sigma^i_1 \ldots \sigma^i_{n_i}$, where $\req_i$ has $n_i$
requests. For any sequence for a single core $\req_i$,
$\overline{\req_i} = \overline{\req_i[1]}, \ldots, \overline{\req_i[n_i]}$. For
any \mc sequence \req,
$\overline{\req} = \{\overline{\req_1}, \ldots, \overline{\req_p}\}$.  For any
sequence $\req_i$, let $\req_i[j_1,j_2]$ denote the (contiguous) subsequence of
requests $\sigma_{i, j_1}, \ldots, \sigma_{i, j_2}$. Also, we use
$\req_\alpha \cdot \req_\gamma$ to denote the concatenation of two sequences
$\req_\alpha, \req_\gamma$.

We now extend a lemma from~\cite{AngeSc09} about sequences with locality that we
will use in our main theorem later. The lemma says that if a sequence
$\ldots\delta\ldots\beta\ldots\delta\ldots\beta\ldots$ exhibits locality of
reference, then $\ldots \delta \ldots \beta \ldots \beta \ldots \delta$ does as
well.

\begin{restatable}{lemma}{localorder}
  \label{lem:local-order}
  Let \req be a sequence of requests consistent with $f$, \ma be a \paging
  algorithm, and \nar be the time that it takes \alg to serve
  \req. Let $j \leq \nar$ be an (integer) \ts such that $\skedar[1,j]$
  contains a request to $\beta$, and in addition, $\delta$ does not appear in
  $\skedpre = \sked_{\req, \ma}[1,j]$ after the last request to $\beta$ in
  $\skedpre$.

  Let $\req' = \reqpre\overline{\reqsuf}$ denote the sequence
  $\req^{\leq j, \ma}\overline{\req^{>j, \ma}}$, and suppose that $\req'$ is not
  consistent with $f$. Then \reqsuf contains a request to $\beta$; furthermore,
  no request to $\delta$ in \skedsuf ($\skedsuf = \skedar[j+1,\nar])$ occurs
  earlier than the first request to $\beta$ in \skedsuf.
\end{restatable}

The following lemma guarantees that for any algorithm \ma which may make a
non-\lru -like eviction at the $(j+1)$-th timestep of some $\req \in \seqset^f$
(but will make \lru -like evictions for the rest of the timesteps after $j+1$),
we can define an algorithm \mb that makes the same decisions as \ma up until
timestep $j$ of any sequence in $\seqset^f$, makes an \lru -like decision on the
$(j+1)$-th timestep, and is no worse than \ma under surjective analysis.

  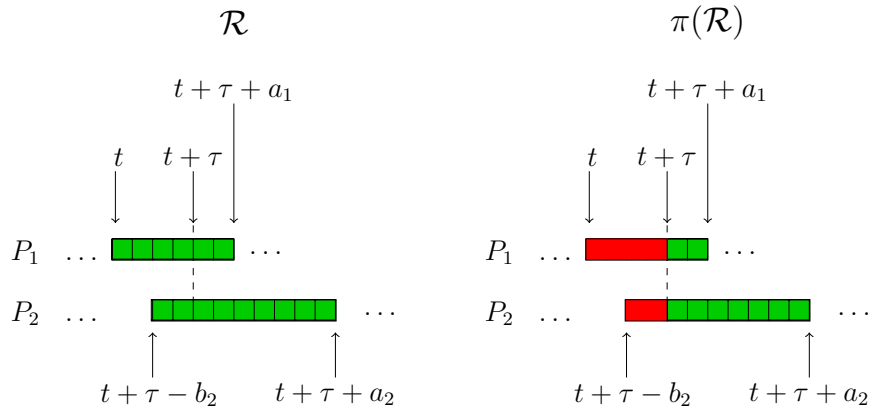
\begin{figure*}[t]
    \centering
    \scalebox{0.9}{
  \begin{tikzpicture}[shorten >=1pt,node distance=1.0cm,on
    grid,auto] 
\newlength{\shft}
\setlength{\shft}{7cm}
     \Large
    \node at (3,5) {\req };
    \large
	\node at (3,4) {$t+\tau+a_1$}; \draw[->] (3,3.8) -- (3,2);
	\node at (2.4,3) {$t+\tau$}; \draw[->] (2.4,2.8) -- (2.4,2); \draw[-,dashed] (2.4,2) -- (2.4,.6);
	\node at (1.3,3) {$t$}; \draw[->] (1.25,2.8) -- (1.25,2);
	\node at (.35,1.6) {$\core_1 \hspace*{3mm}\ldots$ };	\node at (3.4,1.6) {$\hspace*{1.4mm}\ldots$ };
	\node at (.35,.7) {$\core_2 \hspace*{3mm}\ldots$ };    \node at (5,.7) {$\hspace*{3.3mm}\ldots$ };
	\filldraw[step=3mm,fill=black!20!green](1.2,1.8) grid (3,1.485)rectangle (1.2,1.8);
	\filldraw[step=3mm,fill=black!20!green](1.78,.9) grid (4.51,.585)rectangle (1.78,.9);
	\node at (1.9,-.5) {$t+\tau-b_2$}; \draw[->] (1.8,-.25) -- (1.8,0.45);
	\node at (4.5,-.5) {$t+\tau+a_2$}; \draw[->] (4.5,-.25) -- (4.5,0.45);

	\begin{scope}[shift={(7,0)}]
     \Large
    \node at (3,5) {$\pi(\req)$}; 
    \large
	\node at (3,4) {$t+\tau+a_1$}; \draw[->] (3,3.8) -- (3,2);
	\node at (2.4,3) {$t+\tau$}; \draw[->] (2.4,2.8) -- (2.4,2); \draw[-,dashed] (2.4,2) -- (2.4,.6);
	\node at (1.3,3) {$t$}; \draw[->] (1.25,2.8) -- (1.25,2);
	\node at (.35,1.6) {$\core_1 \hspace*{3mm}\ldots$ };	\node at (3.4,1.6) {$\hspace*{1.4mm}\ldots$ };
	\node at (.35,.7) {$\core_2 \hspace*{3mm}\ldots$ };    \node at (5,.7) {$\hspace*{3.3mm}\ldots$ };
	\filldraw[step=3mm,fill=black!20!green](1.2,1.8) grid (3,1.485)rectangle (2.4,1.8);\draw[fill=red]  (2.4,1.8) rectangle(1.2,1.485);
	\filldraw[step=3mm,fill=black!20!green](1.78,.9) grid (4.51,.585)rectangle (2.4,.9) ;\draw[fill=red]  (2.4,.585) rectangle(1.78,.9);
	\node at (1.9,-.5) {$t+\tau-b_2$}; \draw[->] (1.8,-.25) -- (1.8,0.45);
	\node at (4.5,-.5) {$t+\tau+a_2$}; \draw[->] (4.5,-.25) -- (4.5,0.45);
	\end{scope}
      \end{tikzpicture}
      }
      \caption{An example of the mapping of an \inp $\req$ under algorithm \mb
        to $\pi(\req)$ under with $\tau = 4$. On the left, an \inp $\req$ where
        $a_1 = 2, b_2 = 2, a_2 = 7$. The green boxes indicate hits on a page
        $\sigma$ in \mb's cache but not in algorithm \ma's cache. On the right,
        we show the corresponding $\pi(\req)$. The red boxes denote misses on
        $\sigma$. }
 \label{fig:invpage}
\end{figure*}

\begin{restatable}{lemma}{lrusurj}
  \label{lem:lru-surjection}
  Let $\localset$ be \emph{all} \inps consistent with $f$ and let $j$ be an
  integer.  Suppose \ma is an algorithm with the property that for every \inp
  $\req \in \seqset^f$, \alg is \lru -like on timestep $t+1$, for all
  $t \geq j + 1$. Then there exists an algorithm \mb with the following
  properties:
  \begin{enumerate}
    \item For every input $\req \in \seqset^f$, \mb makes the same decisions as
      \ma on the first $j$ timesteps while serving $\req$ (i.e.\ , \ma and \mb
      make the same eviction decisions for each miss in requests up to and
      including time $t$).
    \item For every input $\req \in \seqset^f$, \mb is \lru -like on \req at
      timestep $t$.
    \item $\mb \sleq^f \ma$.
  \end{enumerate}
\end{restatable}

\begin{proofsketch}
  The main insight in this proof is the comparison of \inps with different
  numbers of page requests but the same cost under two different algorithms. If
  an algorithm \ma makes a non-\lru -like decision at some timestep, we
  construct a surjection that maps it to a sequence with the same schedule under
  another algorithm \mb.

  At a high level, we use a ``sequence reordering'' mapping inspired by Lemma 2
  of~\cite{AngeSc09}. Let \mb be an algorithm that matches the evictions of \ma
  until time $t$, when it makes \lru -like evictions. Suppose at time $t$ that
  \ma evicted a page \pnlru and \mb evicted a page \plru. We construct \mb to
  evict the same pages as \ma on the remainder of the sequence.

  We construct a surjective mapping $\pi$ such that for any request
  sequence $\req$, $\mb(\req) \leq \ma(\pi(\req))$.
  There are two main cases based on the continuation of the \inp after time
  $t$. At a high level, if an \inp has locality of reference, then there are not
  many requests to different pages. 
  Now, if possible, we
  swap $\plru, \pnlru$ in the continuation of the \inp after time $t$ since
  these will result in the same cost in the continuation.

  \textbf{Case 1:} Swapping $\plru, \pnlru$ in the continuation maintains
    locality. In this case, $\ma(\req) = \mb(\pi(\req))$ because the different
  decisions at time $t$ did not affect the number of misses (and therefore the
  total time) while serving the the rest of the \inp. Swapping the pages where
  \ma, \mb differ in the continuation of the mapped-to \inp results in the same
  behavior.

  \textbf{Case 2:} Swapping $\plru, \pnlru$ in the continuation does not
  maintain locality. There are a few cases when swapping the two pages would
  disrupt locality.

  \begin{itemize}[leftmargin=*]
  \item If there was a miss on another page before the first request to $\pnlru$
    in the continuation after time $t$, both algorithms would incur the same
    cost since the difference in decision does not affect the number of hits and
    misses in the rest of the \inp. In this case, we set $\pi(\req) = \req$, and
    $\mb(\req) = \ma(\pi(\req))$.
  \item If there was not a miss before the first request to $\pnlru$ after time
    $t$, \mb hits on the first request to $\plru$ in the continuation, and
    we remove requests in $\pi(\req)$ so that the schedule of $\mb$ serving
    $\req$ matches the schedule of $\ma$ serving $\pi(\req)$. Since the
    schedules match, $\mb(\req) = \ma(\pi(\req))$.
  \item The above two cases cover the entire codomain, but not the domain. For
    the remaining \inps, we can map them arbitrarily to \inps of higher cost
    such that there are no more than two \inps in the domain mapped to any \inp
    in the codomain. By construction, $\mb(\req) < \ma(\pi(\req))$. We present
    an example of generating such a mapping from an input \req under algorithms
    \ma and \mb given page $\sigma$ in~\figref{invpage}.
  \end{itemize}\vspace*{-3mm}
\end{proofsketch}

Given any algorithm \ma, we repeatedly apply~\lemref{lru-surjection} to
construct a new algorithm \mb which is \lru -like after some timestep $t$ and is
no worse than \ma.

  Let \nar be the time it takes to serve \inp \req with \ma, and let
  $B_t$ be the class of algorithms that make \lru -like decisions on
  timesteps $n_\req - t$ of every \inp $\req \in \localset$.

   \begin{restatable}{lemma}{lrusteps}
    \label{lem:lru-steps}
    For every algorithm \ma there exists an
    algorithm $\mb_t \in B_t$ such that $\mb_t \sleq \ma$, and for every
    \inp $\req \in \localset$, $\mb_t$ makes the same decisions as
    \ma during the first $\nar - t$ timesteps while serving \req.
  \end{restatable}

  For every \lm algorithm \ma,~\lemref{lru-steps} guarantees the existence of an
  algorithm $\mb$ that makes \lru -like decisions on all timesteps for any \inp
  in \localset and is no worse than \ma. The only algorithm with this property
  is exactly \lru.

\begin{theorem}
  \label{lem:lru-sep-surj}
  For any \lm \paging algorithm \ma, $\lru \sleq^f \ma$.
\end{theorem}

We have defined a surjection from $\lru$ to any other algorithm through
intermediate algorithms that are progressively ``closer to \lru''. Therefore, we
have shown that \lru is the best \lm algorithm under \shortnewba via surjective
analysis and therefore under \shortnewba by combining
Theorem~\ref{lem:lru-sep-surj} and \lemref{surj-bij-equiv}.

\begin{theorem}
  \label{thm:lru-sep}
  For any lazy algorithm \ma, $\lru \bless^f \mb$.
\end{theorem}

We take the first steps beyond worst-case analysis for \mcp with the separation
of \lru from all other \lazy algorithms on \inps with locality via
\shortnewba. The main insight in the proof is to compare \inps of different
lengths (in terms of the number of page requests) but the same schedule with a
surjective mapping and then to convert the mapping into a bijection. Although we
used it the case of \mcp, \fullnewba is a general analysis technique that may
be applied to other online problems.


\section{Related \mcp models}\label{sec:related}
We review alternative models for \mcp in order to explain why we use the
free-interleaving model. Specifically, we discuss a class of models for \mc
\paging called fixed interleaving and the \hass model introduced by
Hassidim~\cite{Hass10}. At a high level, these models assume the order in which
the requests are served is decided by the adversary.  In practice, however, the
schedule of an algorithm is implicitly defined through the eviction
strategies~\cite{LopezSa12a,LopezSa12b}, so the free-interleaving model studied
in this paper is more practical.

Existing work focuses on minimizing either the makespan of \paging strategies or
on minimizing the number of misses.  In the case of single-core \paging,
minimizing the makespan and number of misses are equivalent as makespan is
simply $\tau$ times the number of misses.  For \mcp, however, there is no such
direct relationship between makespan and number of misses. In this paper, we
introduce the total time, a cost measure with benefits over both makespan and
number of misses while capturing aspects of each.

Feuerstein and Strejilevich de Loma~\cite{Stre98,FeuerSt02} introduced
multi-threaded \paging as the problem of determining an optimal schedule in terms
of the optimal interleaved request sequence from a set of individual request
sequences from multiple cores. More precisely, given $p$ request sequences
$\req_1, \ldots, \req_p$, they study miss and makespan minimization for a
``flattened'' interleaving of all $\req_i$'s.  Our work focuses on algorithms
for page replacement rather than ordering (scheduling) of the input
sequences. As mentioned, in practice, the schedule of page requests is embedded
in the page-replacement algorithm.

Several previous works~\cite{BarveGrVi00,CaoFeLi94,KattiRa12} studied \mcp in
the \defn{fixed-interleaving model} (named by Katti and
Ramachandran~\cite{KattiRa12}). This model assumes each core has full knowledge
of its future request sequence where the offline algorithm has knowledge of the
interleaving of requests. The interleaving of requests among cores is the same
for all \paging algorithms and potentially adversarial (for competitive
analysis).  Katti and Ramachandran~\cite{KattiRa12} gave lower bounds and a
competitive algorithm for fixed interleaving with cores that have full knowledge
of their individual request sequences.
In practice, cores 
do not have any knowledge about future requests, and do not necessarily serve
requests at the same rate. Instead, they serve requests at different rates
depending on whether they need to fetch pages to the cache.




Hassidim~\cite{Hass10} introduced a model for \mcp before free interleaving
which we call the \defn{\hass model} that allows offline algorithms to define an
explicit schedule (ordering of requests) for the online algorithm. Given an
explicit schedule, the online algorithm serves an interleaved sequence in the
same way that a single-core algorithm does. The cost of the algorithm, measured
in terms of makespan, is then compared against the cost of an optimal offline
algorithm (which potentially serves the input using another schedule).

Both \hass and free-interleaving models include a fetch delay upon a miss, but
\hass gives offline algorithms more power by allowing them to arbitrarily delay
the start of sequences at no cost in terms of the number of misses (Theorem 3.1
of~\cite{Hass10}). While \hass provides useful insight about serving multiple
request sequences simultaneously, it leads to overly pessimistic results when
minimizing the number of misses as 
it gives offline algorithms an unfair advantage.

Finally, competitive analysis for distributed systems illustrates the difficulty
of multiple independent processes. For example, system nondeterminism in
distributed algorithms~\cite{Aspnes98} addresses nondeterminism in the system as
well as in the input. Furthermore, recent work~\cite{BoyarElLa20} confirms the
difficulty that online algorithms face in ``scheduling'' multiple inputs in the
distributed setting.





\section{Conclusions}\label{sec:conclusion}

We take the first steps beyond worst-case analysis of \mcp in this
paper. In~\thmref{lru-sep}, we separated \lru from other algorithms on sequences
with locality of reference. More generally, we introduced \fullnewba and
demonstrated its flexibility in the direct comparison of online algorithms.  We
expect \fullnewba to be useful in the study of other online problems, and leave
such application as future work.

We conclude by explaining that we are optimistic about \mcp. \MCp is an
important problem in online algorithms and motivated by computer architectures
with hierarchical memory.  Practitioners have extensively studied
cache-replacement policies for multiple cores. The need for theoretical
understanding of \mcp will only grow as \mc architectures become more
prevalent.



\section*{Acknowledgments}
Research was sponsored by the United States Air Force Research Laboratory and
was accomplished under Cooperative Agreement Number FA8750-19-2-1000. The views
and conclusions contained in this document are those of the authors and should
not be interpreted as representing the official policies, either expressed or
implied, of the United States Air Force or the U.S. Government. The
U.S. Government is authorized to reproduce and distribute reprints for
Government purposes notwithstanding any copyright notation herein.

\clearpage
 \bibliographystyle{siam}
 \bibliography{extended-paging}

\begin{thebibliography}{10}

\bibitem{AgraBeDa20}
{\sc K.~Agrawal, M.~A. Bender, R.~Das, W.~Kuszmaul, E.~Peserico, and
  M.~Scquizzato}, {\em Green paging and parallel paging}, in Proceedings of the
  32nd ACM Symposium on Parallelism in Algorithms and Architectures, SPAA '20,
  New York, NY, USA, 2020, Association for Computing Machinery, p.~493–495.

\bibitem{AgrawalBDKPS21}
{\sc K.~Agrawal, M.~A. Bender, R.~Das, W.~Kuszmaul, E.~Peserico, and
  M.~Scquizzato}, {\em Tight bounds for parallel paging and green paging}, in
  Proceedings of the Thirty-Second Symposium on Discrete Algorithms (SODA),
  SIAM, 2021.

\bibitem{AlbersFaGi02}
{\sc S.~Albers, L.~M. Favrholdt, and O.~Giel}, {\em On paging with locality of
  reference}, in Proceedings of the thiry-fourth annual ACM symposium on Theory
  of computing, ACM, 2002, pp.~258--267.

\bibitem{AngeDoLo07}
{\sc S.~Angelopoulos, R.~Dorrigiv, and A.~L{\'o}pez-Ortiz}, {\em On the
  separation and equivalence of paging strategies}, in Proceedings of the
  Eighteenth Annual {{ACM}}-{{SIAM}} Symposium on {{Discrete}} Algorithms,
  {Society for Industrial and Applied Mathematics}, 2007, pp.~229--237.

\bibitem{AngeDoLo08}
\leavevmode\vrule height 2pt depth -1.6pt width 23pt, {\em List update with
  locality of reference}, in Latin American Symposium on Theoretical
  Informatics, Springer, 2008, pp.~399--410.

\bibitem{AngeReSc20}
{\sc S.~Angelopoulos, M.~P. Renault, and P.~Schweitzer}, {\em Stochastic
  dominance and the bijective ratio of online algorithms}, Algorithmica, 82
  (2020), pp.~1101--1135.

\bibitem{AngeSc09}
{\sc S.~Angelopoulos and P.~Schweitzer}, {\em Paging and list update under
  bijective analysis}, in Proceedings of the twentieth annual ACM-SIAM
  symposium on discrete algorithms, SIAM, 2009, pp.~1136--1145.

\bibitem{AngeSc13}
\leavevmode\vrule height 2pt depth -1.6pt width 23pt, {\em Paging and list
  update under bijective analysis}, Journal of the ACM (JACM), 60 (2013),
  pp.~1--18.

\bibitem{Aspnes98}
{\sc J.~Aspnes}, {\em Competitive analysis of distributed algorithms}, in
  Online Algorithms, Springer, 1998, pp.~118--146.

\bibitem{BarveGrVi00}
{\sc R.~D. Barve, E.~F. Grove, and J.~S. Vitter}, {\em Application-controlled
  paging for a shared cache}, SIAM Journal on Computing, 29 (2000),
  pp.~1290--1303.

\bibitem{Belady66}
{\sc L.~A. Belady}, {\em A study of replacement algorithms for a
  virtual-storage computer}, IBM Systems journal, 5 (1966), pp.~78--101.

\bibitem{BenBo94}
{\sc S.~Ben-David and A.~Borodin}, {\em A new measure for the study of on-line
  algorithms}, Algorithmica, 11 (1994), pp.~73--91.

\bibitem{BoyarElLa20}
{\sc J.~Boyar, F.~Ellen, and K.~S. Larsen}, {\em Randomized distributed online
  algorithms against adaptive offline adversaries}, Information Processing
  Letters,  (2020), p.~105973.

\bibitem{BoyarFaLe05}
{\sc J.~Boyar, L.~M. Favrholdt, and K.~S. Larsen}, {\em The relative worst
  order ratio applied to paging}, in Proceedings of SODA, 2005, pp.~718--727.

\bibitem{BoyarFaLa18}
\leavevmode\vrule height 2pt depth -1.6pt width 23pt, {\em Relative worst-order
  analysis: A survey}, in Adventures Between Lower Bounds and Higher Altitudes,
  Springer, 2018, pp.~216--230.

\bibitem{BoyarLaNi01}
{\sc J.~Boyar, K.~S. Larsen, and M.~N. Nielsen}, {\em The accommodating
  function: A generalization of the competitive ratio}, SIAM Journal on
  Computing, 31 (2001), pp.~233--258.

\bibitem{CaoFeLi94}
{\sc P.~Cao, E.~W. Felten, and K.~Li}, {\em Application-controlled {{File
  Caching Policies}}}, in Proceedings of the {{USENIX Summer}} 1994 {{Technical
  Conference}} on {{USENIX Summer}} 1994 {{Technical Conference}} - {{Volume}}
  1, USTC'94, Boston, Massachusetts, 1994, {USENIX Association}, pp.~11--11.

\bibitem{ChroNo99}
{\sc M.~Chrobak and J.~Noga}, {\em Lru is better than fifo}, Algorithmica, 23
  (1999), pp.~180--185.

\bibitem{Denning68}
{\sc P.~J. Denning}, {\em The working set model for program behavior},
  Communications of the ACM, 11 (1968), pp.~323--333.

\bibitem{Dorr10}
{\sc R.~Dorrigiv}, {\em Alternative measures for the analysis of online
  algorithms}, PhD thesis, 2010.

\bibitem{DorrigivLo05}
{\sc R.~Dorrigiv and A.~L\'{o}pez-Ortiz}, {\em A survey of performance measures
  for on-line algorithms}, SIGACT News, 36 (2005), pp.~67--81.

\bibitem{FedoSeSm06}
{\sc A.~Fedorova, M.~I. Seltzer, and M.~D. Smith}, {\em Cache-fair thread
  scheduling for multicore processors}, tech. rep., Harvard University, 2006.

\bibitem{FeuerSt02}
{\sc E.~Feuerstein and A.~Strejilevich~de Loma}, {\em On-{{Line
  Multi}}-{{Threaded Paging}}}, Algorithmica, 32 (2002), pp.~36--60.

\bibitem{Hass10}
{\sc A.~Hassidim}, {\em Cache {{Replacement Policies}} for {{Multicore
  Processors}}.}, in {{ICS}}, 2010, pp.~501--509.

\bibitem{KamaliXu20}
{\sc S.~Kamali and H.~Xu}, {\em Multicore paging algorithms cannot be
  competitive}, in Proceedings of the 32nd ACM Symposium on Parallelism in
  Algorithms and Architectures, SPAA ’20, New York, NY, USA, 2020,
  Association for Computing Machinery, p.~547–549.

\bibitem{KarlinPhRa92}
{\sc A.~R. Karlin, S.~J. Phillips, and P.~Raghavan}, {\em Markov paging}, in
  Foundations of Computer Science, 1992. Proceedings., 33rd Annual Symposium
  on, IEEE, 1992, pp.~208--217.

\bibitem{KattiRa12}
{\sc A.~K. Katti and V.~Ramachandran}, {\em Competitive cache replacement
  strategies for shared cache environments}, in Parallel \& {{Distributed
  Processing Symposium}} ({{IPDPS}}), 2012 {{IEEE}} 26th {{International}},
  {IEEE}, 2012, pp.~215--226.

\bibitem{Komm16}
{\sc D.~Komm}, {\em Introduction to Online Computation}, Springer, 2016.

\bibitem{KoutsoupiasPa00}
{\sc E.~Koutsoupias and C.~H. Papadimitriou}, {\em Beyond competitive
  analysis}, SIAM Journal on Computing, 30 (2000), pp.~300--317.

\bibitem{LopezSa12b}
{\sc A.~L{\'o}pez-Ortiz and A.~Salinger}, {\em Minimizing cache usage in
  paging}, in International {{Workshop}} on {{Approximation}} and {{Online
  Algorithms}}, {Springer}, 2012, pp.~145--158.

\bibitem{LopezSa12a}
\leavevmode\vrule height 2pt depth -1.6pt width 23pt, {\em Paging for
  multi-core shared caches}, in Proceedings of the 3rd {{Innovations}} in
  {{Theoretical Computer Science Conference}}, {ACM}, 2012, pp.~113--127.

\bibitem{ManaMcSl90}
{\sc M.~S. Manasse, L.~A. McGeoch, and D.~D. Sleator}, {\em Competitive
  algorithms for server problems}, Journal of Algorithms, 11 (1990),
  pp.~208--230.

\bibitem{QureJaPa07}
{\sc M.~K. Qureshi, A.~Jaleel, Y.~N. Patt, S.~C. Steely, and J.~Emer}, {\em
  Adaptive insertion policies for high performance caching}, in {{ACM SIGARCH
  Computer Architecture News}}, vol.~35, 2007, pp.~381--391.

\bibitem{QurePa06}
{\sc M.~K. Qureshi and Y.~N. Patt}, {\em Utility-based cache partitioning:
  {{A}} low-overhead, high-performance, runtime mechanism to partition shared
  caches}, in Microarchitecture,{{MICRO}}. {{ACM International Symposium}} On,
  2006, pp.~423--432.

\bibitem{Salt74}
{\sc J.~H. Saltzer}, {\em A simple linear model of demand paging performance},
  Communications of the ACM, 17 (1974), pp.~181--186.

\bibitem{SilbGaGa14}
{\sc A.~Silberschatz, P.~B. Galvin, and G.~Gagne}, {\em Operating system
  concepts essentials}, John Wiley \& Sons, Inc., 2014.

\bibitem{SleaTa85}
{\sc D.~D. Sleator and R.~E. Tarjan}, {\em Amortized efficiency of list update
  and paging rules}, Communications of the ACM, 28 (1985), pp.~202--208.

\bibitem{StoneTuWo92}
{\sc H.~S. Stone, J.~Turek, and J.~L. Wolf}, {\em Optimal partitioning of cache
  memory}, IEEE Transactions on computers, 41 (1992), pp.~1054--1068.

\bibitem{Stre98}
{\sc A.~{Strejilevich de Loma}}, {\em New {{Results}} on {{Fair Multi}}
  threaded {{Paging}}}, Electronic Journal of SADIO, 1 (1998), pp.~21--36.

\bibitem{SuhRuDe04}
{\sc G.~E. Suh, L.~Rudolph, and S.~Devadas}, {\em Dynamic partitioning of
  shared cache memory}, The Journal of Supercomputing, 28 (2004), pp.~7--26.

\bibitem{XieLo09}
{\sc Y.~Xie and G.~H. Loh}, {\em {{PIPP}}: Promotion/insertion
  pseudo-partitioning of multi-core shared caches}, in {{ACM SIGARCH Computer
  Architecture News}}, vol.~37, {ACM}, 2009, pp.~174--183.

\bibitem{Young94}
{\sc N.~E. Young}, {\em The k-server dual and loose competitiveness for
  paging}, Algorithmica, 11 (1994), pp.~525--541.

\bibitem{Young98}
\leavevmode\vrule height 2pt depth -1.6pt width 23pt, {\em Bounding the diffuse
  adversary.}, in SODA, vol.~98, 1998, pp.~420--425.

\bibitem{Young00}
{\sc N.~E. Young}, {\em On-line paging against adversarially biased random
  inputs}, J. Algorithms, 37 (2000), pp.~218--235.

\bibitem{Young02}
{\sc N.~E. Young}, {\em Online file caching}, Algorithmica, 33 (2002),
  pp.~371--383.

\end{thebibliography}

 \clearpage
\appendix
\section{Equivalence of \lm algorithms (from ~\secref{bijective})}\label{app:bijective}


\begin{prop}
  \label{propo:lazybijective}
  If \ma and \mb are two arbitrary \lm algorithms, $\ma \bequiv \mb$.
\end{prop}
\begin{proof}
  The proof is an extension of the proof of Theorem 3.3 from~\cite{AngeDoLo08}.
  Let $(n^t_1, \ldots, n^t_p)$ be the indices of $\req_1, \ldots, \req_p$ being
  served at time $t$ by $\ma, \mb$. Let $n^t = \sum_{i=1}^p n^t_i$ be the number
  of requests served up until time $t$.

  We prove by induction on time that for every $t \geq 1$ that there is a
  bijection $b^t: \seqsetnt \leftrightarrow \seqsetnt$ such that
  $\ma(\req) = \mb(b^t(\req))$ for each $\req \in \seqsetnt$. For
  $t \leq k\tau/p$, $\ma(\req) = \mb(b^t(\req))$ trivially because \ma and \mb
  can only bring in up to $k$ pages, so \ma and \mb behave the same and incur
  the same cost.  Assume that for all $n^t \leq h$ where $h \geq k / p$, we can
  define a bijection $b^t: \req(n^t_1, \ldots, n^t_p)$ showing \ma and \mb are
  equivalent, where $n^t_i$ is the number of requests up to time $t$ of
  core $\core_i$.  We now show how to extend this bijection for $n = h + 1$. We
  define a new bijection \newline
  $b^{h+1} : \seqset(n^{h+1}_1, \ldots, n^{h+1}_p) \leftrightarrow
  \seqset(n^{h+1}_1, \ldots, n^{h+1}_p)$, which maps the continuations of each
  request sequence $\req_i$ to the continuations of $b^h(\req_i)$ in the image.
  By assumption, up to time $h$ we have defined a bijection $b^h$ that matches
  sequences for \ma, \mb in terms of cost and schedule. That is, the number of
  pages $k' < k$ being fetched at time $t$ after serving \req by \ma is the same
  as the number of pages being fetched at time $t$ after serving $b^h(\req)$ by
  \mb.

  Let $|P| = N$ be the number of distinct pages that any algorithm can request.

  Since there are i) $k-k'$ possible next-hit requests in both \ma and \mb at
  time $h+1$ and ii) the same number of cores not currently fetching in
  $\ma, \mb$ at time $h+1$, we can arbitrarily biject these to each other in
  each $\req_i$. We also do the same for the $N-k$ next-miss requests outside
  the cache and the misses on the $k'$ requests being fetched for each
  $\req_i$. \ma and \mb incur the same cost in each mapping and maintain the
  same schedule, $\ma \equiv_b \mb$.
\end{proof}

\clearpage
\section{Proofs for Lemmas in~\secref{surjective}}\label{app:surj-lemmas}

\subsection{Formalizing \lru in the \mc setting}\label{sec:formal-lru}

In order to compare algorithms with \lru, we compare the state of the cache and
the timestamps assigned to pages in the cache throughout the execution of
different algorithms. At each timestep, \lru assigns integer
\defn{tags}~\cite{AngeSc09} to each page in its cache to represent when they
were most-recently accessed.

In general, an algorithm \ma is \defn{tag-based} if it uses tags to keep track
of when pages were last accessed. Given an algorithm $\ma$ that uses tags, we
denote the tag of some page $\sigma$ in the cache at time $t$ with
$\text{tag}_\ma[\sigma, \mathcal{S}_t]$, where $\mathcal{S}_t$ is the schedule
of the \inp up to time $t$.

Since will be comparing \lru with arbitrary algorithms via surjective analysis,
we will formalize \defn{tag-based \lru}~\cite{AngeSc09} in a shared cache.
Tag-based \lru in the \mc setting is a straightforward extension of its
definition in the single-core setting.

\begin{definition}[Tag-based \lru (\cite{AngeSc09})]
  \label{def:tag-based}
  \defn{Tag-based} \lru assigns a set $T$ of (integer) \emph{tags} to each page
  in its cache to represent when they were most-recently accessed. Formally, for
  every page $\sigma$ in the cache, let $\text{tag}_\lru[\sigma, \skedar[t]]$ be
  the tag assigned to $\sigma$ right after \lru has served requests up to
  timestep $t$. Tag-based \lru processes each
  request $\sigma$ at each timestep $\ell > t$ as follows:
  \begin{enumerate}
  \item If $\sigma$ is a hit, \lru updates the tag of
    $\text{tag}_\lru[\sigma, \skedar[\ell]] = \ell$.
 \item If $\sigma$ is a miss and not currently being fetched by another core,
   \lru will evict the page with the smallest tag (if the cache is full) and
   fetch $\sigma$ to the cache while updating its tag for the next $\tau$
   timesteps as it is fetched.
 \item If $\sigma$ is a miss and currently being fetched by another core, \lru
   will not evict a page (since the eviction due to $\sigma$ already happened)
   and the core that requested $\sigma$ will stall for $x$ steps until $\sigma$
   is brought to the cache.
 \end{enumerate}
 \end{definition}

\subsection{Proofs of Lemmas}
\localorder*

\begin{proof}
Since $\req' = \reqpre\overline{\reqsuf}$ is not consistent with $f$, there must
exist indices $j_{1,1}, j_{1, 2}, \ldots, j_{p, 1}, j_{p, 2}$ such that for all
$i = 1, \ldots, p$, $j_{i,1} < j_{i, 2} \leq n_i$ such that the number of
distinct requests over all $\req_i[j_{i,1}, j_{i,2}]$ exceeds $f(j_2 - j_1 + 1)$
distinct pages. For any subsequence $r$ in $\reqsuf$, $\overline{\req}$ has the
same number of distinct pages as $r$. Therefore, at least one of
$j_{i,1}, j_{i, 2}$ must be such that $j_{i, 1} \leq t_{i, j} \leq j_{i, 2}$
(where $t_{i,j}$ is the index of some $\req_i$ at time $j$ under \ma.

Suffices then to argue that for at least one $i = 1, \ldots, p$,
$\req_i[t_{i,j}, j_{i,2}]$ contains a request to $\beta$ but not to
$\delta$. For simplicity, we will specify a subsequence of one $\req_i$ to mean
over all $i = 1, \ldots, p$.

It is easy to see that $\req_i[t_{i,j}, j_{i,2}]$ cannot contain requests to
both $\beta$ and $\delta$, nor can it contain requests to none of these pages:
if either of these cases occurred, then $\req_i[j_{i,1}, j_{i,2}]$ and
$\req_i[j_{i, 1}, t_{i, j}]\overline{\req_i[t_{i, j}, j_{i, 2}]}$ would contain
the same number of distinct pages, which contradicts that $\req$ is consistent
with $f$. Note that $\req_i[j_{i, 1}, t_{i, j}]$ contains a request to $\beta$
but not to $\delta$.

Now \reqpre$\overline{\reqsuf}$ must contain a request that does not appear in
$\req_i[j_{i,1}, j_{i,2}]$ and $\delta$ is the only option. Therefore, \reqsuf
contains $\beta$ but not $\delta$.
\end{proof}

We advise the reader to first focus on the structure of the proof
of~\lemref{lru-surjection} by skipping the proofs of the propositions, and then
revisiting the details afterwards in~\appref{surj-propos}.

\lrusurj*

\begin{proof}
  First, we construct \mb using \ma on an \inp $\req \in \localset$.
  At a high level, \mb matches \ma 's eviction decisions up to time $j$, makes
  an \lru -like decision at time $j+1$, and matches \ma in the remainder of the
  \inp.  First, we require \mb to make the same decisions as \ma on all requests
  in \reqpre. If \ma makes \lru -like decisions on all misses at time $j+1$,
  then \mb makes the same \lru -like decision as \ma, as well as the same
  decisions on all \reqsuf as \ma.

  If \ma makes a non-\lru -like decision at time $j+1$, however, there must
  exist a pair of pages $\plru, \pnlru \in P$ where $\plru \neq \pnlru$ such
  that at timestep $j+1$, \ma evicts \pnlru from its cache, whereas \plru is the
  least-recently-used page in \reqpre (for now we assume that \ma, \mb differ by
  only one page. The mapping in this lemma can be repeated for multiple pages,
  however.) If there are multiple non-\lru -like decisions at time $j+1$, we can
  apply the same sequence-mapping technique for all of them.

  We require that \mb evicts \plru in the remainder of the \inp if there is a
  miss. The tag of all other pages besides \pnlru is defined by the last time
  there were accessed, and the tag of \pnlru is the last time \plru was
  accessed.  More formally,
  $\text{tag}_\mb[\pnlru, \skedpre \cdot s^\ma] \gets \text{last}[\plru,
  \skedpre]$, and
  $\text{tag}_\mb[\sigma, \skedpre \cdot s^\ma] \gets \text{last}[\sigma,
  \skedpre]$ for all pages $\sigma \neq \pnlru$ in \mb's cache after time
  $j+1$. We use $\text{last}[\plru, \skedpre]$ to denote the \emph{time} of the
  last access to \plru in \skedpre. After time $j+1$, we require that \mb is
  tag-based. Note that \mb is completely online because it does not know the
  future.

  The two algorithms differ in only one eviction: \mb evicts \plru instead of
  \pnlru (makes an \lru -like decision) and demotes the timestamp of \pnlru so
  that \pnlru is the least-recently-used page as \mb prepares to serve the
  suffix \reqsuf.

  By construction, \mb satisfies properties (1) and (2) of the lemma. In the
  rest of the proof, we will show property (3). Let $\skedloc$ be the set of
  schedules resulting from serving \inps with locality \localset with $\ma$.

  We now define a mapping between \inps served by algorithms that differ on
  one eviction such that the two \inps have the same schedule.

  \begin{definition}[Inverse \inp on one page]
    \label{def:inv-seq}
    Let $\sigma$ be a page that algorithm $\mb$ hits on and $\ma$ misses on (for
    the \emph{first time} after time $j+1$) at time $t > j+1$. Also, suppose
    that \req is an \inp with at least $\tau$ repetitions of $\sigma$
    starting at time $t$ under \mb.  We define the \defn{inverse of \req in \mb
      under \ma w.r.t.\ $\sigma$}, \invpage, as as follows: \invpage under \ma
    generates the same schedule as \req under \mb. Informally, \invpage removes
    all repetitions due to misses the first time $\sigma$ is fetched after time
    $j+1$.

    Let \req be an \inp such that at least one core $\core_i$
    requests $\sigma$ at least $\tau$ times starting at timestep $t$ when served
    by \mb. Formally, let $\core_i$ request $\sigma$ $\tau + a_i$ times starting
    at time $t$ under $\mb$, at index $x_i$ through $x_i + \tau + a_i$ in
    $\req_i$. In \invpage, we map those requests to a ``shorter'' \inp of
    repetitions: starting at index $x_i$ in $\req_i$, \invpage only has
    $a_i + 1$ requests to $\sigma$. Furthermore, suppose any other core
    $\core_j \neq \core_i$ repeats requests to $\sigma$ at least $b_j + a_j$
    times starting at some timestep $t + \tau - b_j$ (for $0 < b_j \leq \tau$)
    and that they begin at index $x_j$. We map those requests to $a_i + 1$
    repetitions of $\tau$ in \invpage. Note that for all $i = 1, \ldots, p$,
    $a_i \geq 0$.

    The inverse \invpage is only defined for \inps that have at
    least $\tau$ repetitions of $\sigma$ at time $t$ under \mb.
    Let \invpageset be the set
    of \inps with locality where the inverse is defined for \ma.
  \end{definition}

  We present an example of generating \invpage from \req under \ma and \mb given
  page $\sigma$ in ~\figref{invpage}. In the example, we ``shorten'' the
  repetitions in \invpage such that \ma serving \invpage generates the same
  schedule as \mb serving \req. In \invpage, $p_1$ requests $\sigma$ 3 times
  ($a_1 + 1$) and $p_2$ requests $\sigma$ 8 times ($a_2 + 1$).

  \begin{prop}
  \label{propo:add-reqs-local}
  Let $f$ be an increasing concave function and $\ma$ be any \paging
  algorithm. If an \inp $\req$ is consistent with $f$, an \inp $\req'$
  based on \req that repeats any of its requests $\sigma$ (immediately after
  $\sigma$) is also consistent with $f$.
  \end{prop}
  The only difference between $\req'$ and \req is that $\req'$ may have some
  repeated requests. Repeating requests does not increase the number of distinct
  pages in each window, so $\req'$ must also be consistent with $f$.

  Note that even if an \inp \req has locality of reference and has at least
  $\tau$ repetitions of $\sigma$ at time $t$ under \mb, \invpage may not have
  locality of reference as it removes duplicates. Every local \inp that
  misses on $\sigma$ at time $t$ under \ma has a corresponding \inp with
  repetitions to replicate \ma's schedule under \mb, however, because creating
  the same schedule in \mb requires only adding repetitions, which maintain
  locality (\proporef{add-reqs-local}).

  We use surjective analysis via case analysis of the space of request
  \inps with locality as follows:
{\tiny
  \begin{numcases}{\pi(\req) =}
    \fullcompsuf & if $\reqpre r_{j+1}
    \compsuf$ is consistent with $f$ \notag \\
    & and \notag \ma does not make an \lru -like \notag \\
    & decision on $\reqpre r_{j+1}$. \label{case1}\\
    \req
    & $\reqpre r_{j+1} \compsuf$ is not consistent with $f$,
    \notag \\
    & and \mb incurs a miss before the first\notag \\
    & request to \plru in \reqsuf. \label{case2} \\
    \invpage & $\reqpre r_{j+1} \compsuf$ is not consistent with $f$,
    \notag \\
    & \mb does not incur a miss before the \notag \\
    & first request to \plru in \reqsuf,
    \notag \\
    & and $\req \in  \invpagesetnlru$. \label{case3} \\
    \req' &
    $\reqpre r_{j+1} \compsuf$ is not consistent with $f$,
    \notag \\
    & \mb does not incur a miss before the \notag \\
    & first request to \plru in \reqsuf,
    \notag \\
    & and $\req \notin  \invpagesetnlru$ \label{case4}
  \end{numcases}
  }
  where \compsuf denotes the complement of \reqsuf with respect to \plru and
  \pnlru ($\overline{\reqsuf}^{(\plru, \pnlru)}$). Additionally, $\req'$ is
  another \inp such that \ma serving $\req'$ has a greater total time than
  $\mb$ serving $\req$ (i.e.\ $\mb(\req) < \ma(\req')$).



  First, we show that $\pi(\req)$ accounts for all $\req \in \localset$.

  \begin{restatable}{prop}{surjmap}
    \label{lem:surj-map}
    The function $\pi(\req): \localset \leftrightarrow \localset$ is surjective
    and non-injective.
  \end{restatable}

  \begin{proofsketch}
    Cases 1-3 of $\pi(\req)$ account for the entire codomain but not the entire
    domain, because case (3) is surjective on that partition of the codomain.
    Therefore, $\pi(\req)$ is a natural surjective mapping because there are
    infinitely many \inps in Case~\ref{case4}, so there are infinitely many
    one-to-one mappings in Cases 1-3, and then infinitely many two-to-one
    mappings from Case 4.
  \end{proofsketch}

  Now we will show that for every $\req \in \localset$,
  $\mb(\req) \leq \ma(\pi(\req))$. Again, we only consider the case where \ma
  does not make an \lru -like request at time $j+1$. We proceed by case analysis
  in~\proporeftwo{case1}{case2-4}. Since we will be comparing the cache contents
  of \ma and \mb by induction, we define the cache state of \mb and \ma as
  they serve \req and $\pi(\req)$, respectively.

  \begin{definition}[Cache state (informal, \cite{AngeSc09})]
    The \defn{cache state} of an algorithm \ma at any timestep $t$ consists of
    the set of pages in the cache as well as the tag assigned to each page. For
    a more formal definition, see~\defref{cache-state-formal}.
  \end{definition}

  We choose tags at time $j+1$ to make \ma \lru -like and tag-based on the
  suffix of \req so that we can compare \ma to \mb.

  \begin{restatable}[Case 1 of $\pi(\req)$]{prop}{caseone}
      \label{propo:case1}
      If $\curr \compsuf$ is consistent with $f$, $\mb(\req) = \ma(\pi(\req))$.
    \end{restatable}

    \begin{proofsketch}
      We prove the proposition by induction on the timestep $\ell$. We will show
      that the cache states of \ma and \mb are such that \mb incurs a miss at
      time $\ell$ on \req if and only if \ma incurs a miss at time $\ell$ on
      $\pi(\req)$. We proceed by case analysis.
    \begin{description}
    \item [Case 1.] If none of the requests at time $\ell$ are $\pnlru, \plru$,
      then $\ma, \mb$ have the same behavior and incur the same cost at time
      $\ell$. Therefore, the proposition holds for $\ell + 1$.
    \item[Case 2.] If $\mb$ sees a request to \plru at time $\ell$, then $\ma$
      sees a request to $\pnlru$. By the induction hypothesis, they have the
      same behavior with their respective \pnlru, \plru, and update their
      cache states to assign the same tag to their respective pages.
    \item[Case 3.]  If $\mb$ sees a request to \pnlru at time $\ell$, then $\ma$
      sees a request to $\plru$, and we use a symmetric argument to Case 2.
    \end{description}
  \end{proofsketch}

  \begin{restatable}[Cases 2, 3, 4 of $\pi(\req)$]{prop}{caseothers}
    \label{propo:case2-4}
    If \fullcompsuf is not consistent with $f$, then
    $\mb(\req) \leq \ma(\pi(\req))$.
  \end{restatable}

  \begin{proofsketch}
    We proceed by case analysis on $\pi(\req)$. By construction, \ma and \mb
    incur the same cost up until time $j+1$. Their cache states differ only in
    that \ma's cache contains $\plru$ and \mb's cache contains $\pnlru$.
    Since \fullcompsuf is not consistent with $f$,~\lemref{local-order}
    states that both \plru and \pnlru must appear in the suffix \reqsuf and that
    \pnlru must be requested earlier (in time) than \plru in \reqsuf.
    \begin{description}
    \item[Case 2 of $\pi(\req)$.] If \mb incurs a miss before the first request
      to \plru in \reqsuf, $\pi(\req) = \req$. Both \ma and \mb incur a miss at
      time $\ell$, and replace $\pnlru$ and $\plru$, respectively. Therefore,
      \ma and \mb also have all the same eviction decisions after time $\ell$
      because they have matching cache states, so $\mb(\req) = \ma(\pi(\req))$.
    \item[Case 3 of $\pi(\req)$.] Suppose that the first request to \plru in
      \reqsuf occurs at time $t$.  If $\pi(\req) = \invnlru$, \ma and \mb do not
      incur any misses between times $j+1$ and $t$. At time $t$, \mb incurs a
      hit and \ma incurs a miss. By definition of inverse, \ma and \mb so
      $\mb(\req) = \ma(\pi(\req))$ because they have the same total time
      (repeated requests in \mb to match the miss in \ma).
    \item [Case 4 of $\pi(\req)$.]  If $\pi(\req) = \req'$,
      $\mb(\req) < \ma(\pi(\req))$ by construction of $\req'$.
      \end{description}
    \end{proofsketch}

  We have shown in~\proporeftwo{case1}{case2-4} that there exists a surjection
  $\pi$ such that for all $\req \in \localset$,
  $\mb(\req) \leq \ma(\pi(\req))$.
  \end{proof}

    \lrusteps*

  \begin{proof}
  We proceed by induction on $t$. The lemma is trivially true for $t = 0$. Let
  $\mb_t \in B_t$ be an algorithm such that $\mb_t \sleq \ma$, and for any
  \inp $\req \in \localset$, $\mb_t$ makes the same decisions as \ma
  for the first $\nar - t$ timesteps while serving \req.

  We show that the claim holds for $t+1$ as well. From~\lemref{lru-surjection},
  there exists an algorithm \mb such that $\mb \sleq \mb_t$, and for every
  $\req \in \localset$, \mb makes an \lru -like decision at time $\nar - t$, and
  matches $\mb_t$ on the first $\nar - t - 1$ requests in $\req$.

  Note that \mb does not necessarily make \lru -like decisions for requests
  after $\nar - t + 1$. By the induction hypothesis, there exists an algorithm
  $\mb'_t \in B_t$ such that i) $\mb'_t \sleq \mb$, and ii) for every
  $\req \in \localset$, $\mb'_t$ makes the same decisions as $\mb$ on the first
  $\nar - t$ timesteps of $\req$, and \lru -like decisions on the remaining
  timesteps. By definition, $\mb'_t \in B_{t+1}$. We can reapply the induction
  hypothesis: $\mb'_t$ makes the same decisions as \ma in the first
  $\nar - t - 1$ timesteps of \req, and so the lemma holds for $t+1$.
  \end{proof}

\clearpage
\section{Proofs for Propositions in~\secref{surjective}}\label{app:surj-propos}

\surjmap*

 \begin{proof}
    ~\lemref{lru-surjection} is trivially true if \ma made only \lru -like
    requests at time $j+1$ because \ma and \mb would be the same. Therefore, we
    will consider the case where \ma makes a non-\lru -like eviction at time
    $j+1$.

    We proceed by cases following the definition of $\pi(\req)$.
    \begin{description}
    \item[Case~\ref{case1}.] \ma also does not make an \lru -like eviction at
      time $j+1$ on both \req and $\pi(\req)$. Since the complement of \compsuf
      is just \reqsuf, $\pi(\pi(\req)) = \req$.
    \item[Case~\ref{case2}.] \fullcompsuf is not consistent with $f$ and \mb
      incurs a miss before the first request to \plru in \reqsuf.  Trivially,
      $\pi(\pi(\req)) = \req$ because $\pi(\req) = \req$.
    \item[Case~\ref{case3}.]  If \fullcompsuf is not consistent with $f$, \mb
      does not incur a miss before the first request to \plru in \reqsuf, and
      $\req \in \invpagesetnlru$, then $\pi(\req) = \invnlru$. The set of all
      inverses from $\req \in \invpagesetnlru$ is \emph{all sequences} in \localset
      where \fullcompsuf is not consistent with $f$. From~\defref{inv-seq},
      \invpagesetnlru is the set of all sequences with at least one request to
      $\sigma$ at time $t$.
    \item[Case~\ref{case4}.] If \fullcompsuf is not consistent with $f$, \mb
      does not incur a miss before the first request to \plru in \reqsuf, and
      $\req \notin \invpagesetnlru$, then $\pi(\req) = \req'$.  Cases 1, 2, and 3
      actually map to all of \localset, but we require Case 4 because we have
      not yet accounted for all of the domain. Since we have already defined a
      mapping to all of the codomain in the first three cases, all we need is a
      corresponding input $\req'$ such that $\mb(\req) \leq \ma(\req')$.
    \end{description}

    Therefore, $\pi(\req)$ is a natural surjective mapping because there are
    infinitely many \inps in Case~\ref{case4}, so there are infinitely many
    one-to-one mappings in Cases 1-3, and then infinitely many two-to-one
    mappings from Case 4.
  \end{proof}

  \begin{definition}[Cache state (formal) \cite{AngeSc09}]
    \label{def:cache-state-formal}
    Let $C[\ma, \req]$ be the \defn{cache state} of
    algorithm \ma after it has served \inp $\req$. The cache state consists of
    the set $P[\ma, \req]$ of pages in the cache after serving $\req$, as well
    as assigned tags $\text{tag}_\ma[\sigma, \req]$ equal to
    $\text{last}_\ma[\sigma, \req]$ for all $\sigma \in P[\ma, \req]$.

    For example $C[\ma, \curr]$ is the cache state of \ma after it has served
    requests up to time $j+1$.

    The \defn{complement} of cache state $C[\ma, \req]$ with respect to $\beta$
    and $\delta$, denoted by $\overline{C}[\ma, \req]$ is a cache state in
    which:
    \begin{itemize}
      \item the set of pages is the set $\overline{P[\ma, \req]}$ (where
        $\alpha$ is replaced with $\beta$ and vice versa).
      \item tags are as in $C[\ma, \req]$ except for: if
        $\beta \in \overline{P[\ma, \req]}$ (resp. if
        $\delta \in \overline{P[\ma, \req]}$), then $\beta$'s tag in
        $\overline{C}[\ma, \req]$ is the tag of $\delta$ in $C[\ma, \req]$
        (resp. the tag of $\beta$ in $C[\ma, \req]$).
      \end{itemize}
    \end{definition}

\caseone*

\begin{proof}
  Let $\req^{\leq \ell, \ma}$ be the requests served by $\ma$ up to and
  including time $\ell$.  Let $ms(\ma, \req)$ be the makespan of $\ma$ on
  $\req$. We will show that for all $j+1 \leq \ell \leq ms(\ma, \req)$,
  algorithm \mb satisfies the following properties:
      \begin{enumerate}
      \item
        $C[\mb, \req^{\leq \ell, \ma}] = \overline{C}[\ma, \pi(\req)^{\leq \ell,
          \ma}]$, and
        \item \mb incurs a miss at time $\ell$ on \req if and only if \ma incurs
        a miss at time $\ell$ on $\pi(\req)$.
      \end{enumerate}

      We prove the proposition by induction on the timestep $\ell$. Suppose that
      the claim holds for $\ell < n$: we will show that it holds for $\ell +
      1$. By construction, the claim holds for $\ell = j + 1$; note that the
      actions of \ma on $\pi(\req)$ at time $j+1$ and choice of initial tags
      guarantee that
      $C[\mb, \req^{\leq j+1, \ma}] = \overline{C}[\ma, \pi(\req)^{\leq j+1,
        \ma}]$. We now use case analysis at timestep $\ell + 1$ on requests
      $\skedbri[\ell+1]$ for $i = 1, \ldots, p$ where $\skedbri[\ell+1]$ is the
      request by $p_i$ at time $\ell + 1$ while \mb serves \req. Similarly,
      $\skedmapari[\ell+1]$ for $i = 1, \ldots, p$ is the request by $p_i$ at
      time $\ell + 1$ while \ma serves $\pi(\req)$.
    \begin{description}
    \item [Case 1.] If $\skedbri[\ell+1] \neq \pnlru, \plru$, then \\
      $\pnlru, \plru \neq \skedmapari[\ell+1]$. If a request $\skedbri[\ell+1]$
      is a hit for \mb, it is also a hit for \ma, and both \ma and \mb will
      update the tag of page $\skedbri[\ell+1]$ to $\ell + 1$ in their
      corresponding caches. Similarly, if $\skedbri[\ell+1]$ is a miss for \mb,
      then by the induction hypothesis about the cache configuration of \ma,
      $\overline{\skedbri[\ell+1]}$ will also be a miss in \ma. Additionally,
      \ma and \mb will evict the same page from their cache and update the tag
      of $\skedbri[\ell+1]$ to $\ell + 1$, so the proposition holds for
      $\ell + 1$.
    \item[Case 2.] If $\skedbri[\ell+1] = \plru$, then
      $\skedmapari[\ell + 1] = \pnlru$. We consider two cases: either
      $\skedbri[\ell+1]$ is a hit or miss for \mb. If it was a hit, then by the
      induction hypothesis $\pnlru \in C[\ma, \pi(\req)^{\leq \ell, \ma}]$ and
      $\skedmapari[\ell + 1]$ is a hit in \ma. After serving request
      $\skedbri[\ell+1]$, \mb updates the tag of \plru to $\ell + 1$, and \ma
      sets the tag of \pnlru to $\ell + 1$, so
      $C[\mb, \req^{\leq \ell + 1, \mb}] = \overline{C}[\ma, \pi(\req)^{\leq
        \ell + 1, \ma}]$.  If $\skedbr[\ell+1]_i$ was a miss for \mb, then from
      the induction hypothesis $\skedmapari[\ell + 1]$ was not in \ma's cache at
      time $\ell$. Therefore, \ma and \mb evict the same page in order to bring
      in \plru and \pnlru, respectively, and update the respective tags to
      $\ell + 1$. Therefore, we maintain the invariant that
      $C[\mb, \req^{\leq \ell + 1, \mb}] = \overline{C}[\ma, \pi(\req)^{\leq
        \ell + 1, \ma}]$.
    \item[Case 3.]  If $\skedbri[\ell+1] = \pnlru$, then
      $\skedmapari[\ell + 1] = \plru$. We use a symmetric argument to Case 2.
    \end{description}
  \end{proof}

  \caseothers*

  \begin{proof}
    We proceed by case analysis on $\pi(\req)$. From construction of \mb,
    $\mb(\req^{\leq j+1, \mb}) = \ma(\pi(\req^{\leq j+1, \mb}))$. Additionally,
    from initial choice of tags,
    $C[\mb,\req^{\leq j+1, \mb}] = \overline{C}[\ma, \pi(\req^{\leq j+1,
      \ma})]$. Specifically, \newline
    $C[\mb,\req^{\leq j+1, \mb}], C[\ma,\req^{\leq j+1, \mb}]$ have identical
    page sets, except that the first contains \pnlru and the second contains
    \plru. Since \fullcompsuf is not consistent with $f$, ~\lemref{local-order}
    states that both \plru and \pnlru must appear in the suffix \reqsuf and that
    \pnlru must be requested earlier (in time) than \plru in \reqsuf.
    \begin{description}
    \item[Case 2 of $\pi(\req)$.] If \fullcompsuf is not consistent with $f$ and
      \mb incurs a miss before the first request to \plru in \reqsuf,
      $\pi(\req) = \req$. Suppose that the first request to \plru in \reqsuf
      occurs at timestep $t$ and let $\ell$ ($j + 1 < \ell < t$) be the earliest
      timestep on which \mb incurs a miss before $t$. Let $\sigma^i_\ell$ be
      the page that caused the miss at time $\ell$ requested by $p_i$:
      $\sigma^i_{\ell}$ cannot be \plru.  Every request up to time $\ell$ must
      have been a hit for \mb, and
      $C[\mb,\req^{< \ell, \mb}] = \overline{C}[\ma, \pi(\req)^{< \ell, \ma}]$.
      On request $\sigma^i_\ell$, \mb incurs a miss, evicts \plru (in an \lru
      -like decision), and brings $\sigma^i_{\ell}$ to the cache, and sets its
      tag to $\ell$. Since $\sigma^i_{\ell} \notin \{\plru, \pnlru\}$, \ma will
      also incur a miss in $\pi(\req)$ at time $\ell$ on $\sigma^i_{\ell}$ and
      replace \pnlru with $\sigma^i_{\ell}$ in a tag-based eviction (and also
      set the tag of $\sigma^i_{\ell}$ to $\ell$). Therefore, \ma and \mb have
      all the same eviction decisions after time $\ell$ because
      $C[\mb,\req^{\leq \ell, \mb}] = C[\ma, \pi(\req)^{\leq \ell, \ma}]$, and
      $\mb(\req) = \ma(\pi(\req))$.
    \item[Case 3 of $\pi(\req)$.] Suppose that the first request to \plru in
      \reqsuf occurs at time $t$.  If \fullcompsuf is not consistent with $f$,
      \mb does not incur a miss between times $j+1$ and $t$, and
      $\req \in \invpagesetnlru$, $\pi(\req) = \invnlru$. In this case, \ma also
      does not incur any misses between times $j+1$ and $t$. On request
      $\sigma_i^t = \pnlru$, \mb hits on \pnlru and \ma incurs a miss and makes
      an \lru -like eviction: specifically, it evicts \plru, replaces it with
      \pnlru, and updates its tag to $t + \tau$ (after it is done fetching). At
      time $t + \tau$, the cache states of \ma and \mb are the same
      ($C[\mb, \req^{\leq t+\tau,\mb}] = C[\ma, \req^{\leq t + \tau,
        \ma}]$). Additionally, \mb, \ma are tag-based on each request in
      $\req^{> t + \tau, \mb}, \pi(\req)^{> t + \tau, \ma}$ (which happen to be
      the same). Therefore, the actions of \ma and \mb are the same after time
      $t + \tau$, and so $\mb(\req) = \ma(\pi(\req))$ because they have the same
      total time (repeated requests in \mb to match the miss in \ma).
      \item [Case 4 of $\pi(\req)$.]  If \fullcompsuf is not consistent with
        $f$, \mb does not incur a miss before the first request to \plru in
        \reqsuf, and $\req \notin \invpagesetnlru$, then $\pi(\req) = \req'$. In
        this case, $\mb(\req) < \ma(\pi(\req))$ by construction of $\req'$.
      \end{description}
    \end{proof}


\end{document}